\pdfoutput=1
\documentclass[
paper=letter,%
pagesize,%
numbers=noendperiod,%
captions=nooneline,%
abstracton%
]{scrartcl}

\usepackage[T1]{fontenc}%
\usepackage{lmodern}%
\usepackage[american]{babel}%
\usepackage{microtype}%

\usepackage[
hyperref,%
table%
]{xcolor}%

\usepackage[nouppercase]{scrpage2}%

\usepackage{eso-pic}%
\usepackage{rotating}%
\usepackage{amsmath}%
\usepackage{mathtools}%
\usepackage{amssymb}%
\usepackage{amsthm}%
\usepackage{etoolbox}%
\usepackage{bm}%
\usepackage{bbm}%
\usepackage{enumitem}%
\usepackage{graphicx}%
\usepackage{grffile}%
\usepackage{tikz}%
\usepackage{wrapfig}%
\usepackage{tabularx}%
\usepackage{bbm}
\usepackage{siunitx}%
\usepackage{booktabs}%
\usepackage{multirow}%
\usepackage{fancyvrb}%
\usepackage{listings}%
\usepackage{csquotes}%
\usepackage[
backend=bibtex,%
style=authoryear,%
dashed=false,%
uniquename=init,%
firstinits=true,%
hyperref=true,%
useprefix=true,%
maxcitenames=2,%
maxbibnames=6,%
date=iso8601,%
urldate=iso8601%
]{biblatex}%
\usepackage[
hypertexnames=false,%
setpagesize=false,%
pdfborder={0 0 0},%
pdfstartview=Fit,%
bookmarksopen=true,%
bookmarksnumbered=true%
]{hyperref}%
\xdefinecolor{lightgray}{RGB}{247, 247, 247}%
\xdefinecolor{semilightgray}{RGB}{240, 240, 240}%
\xdefinecolor{middlegray}{RGB}{127, 127, 127}%
\xdefinecolor{blue}{RGB}{58, 95, 205}%
\xdefinecolor{deepskyblue}{RGB}{0, 154, 205}%
\xdefinecolor{chocolate}{RGB}{205, 102, 29}%

\setkomafont{pageheadfoot}{\normalfont\normalcolor\sffamily}%
\setkomafont{pagenumber}{\normalfont\normalcolor\sffamily}%
\automark{section}%
\setkomafont{captionlabel}{\normalfont\normalcolor\sffamily\bfseries}%

\setlist{%
  align=left,%
  labelsep=*,%
  leftmargin=*,%
  topsep=1mm,%
  itemsep=0mm%
}
\newcommand*{\mysquare}{\rule[0.18em]{0.36em}{0.36em}}
\newcommand*{\mytriangle}{\raisebox{0.12em}{\resizebox{0.48em}{0.48em}{$\blacktriangleright$}}}
\newcommand*{\mybar}{\rule[0.32em]{0.62em}{0.08em}}
\newcommand*{\mydot}{\raisebox{0.14em}{\resizebox{0.44em}{!}{$\bullet$}}}
\setlist[itemize,1]{label={\mysquare\ }}%
\setlist[itemize,2]{label={\mytriangle\ }}%
\setlist[itemize,3]{label={\mybar\ }}%
\setlist[itemize,4]{label={\mydot\ }}%
\setlist[enumerate,1]{label=\arabic*)}%
\setlist[enumerate,2]{label=\arabic{enumi}.\arabic*)}%
\setlist[enumerate,3]{label=\arabic{enumi}.\arabic{enumii}.\arabic*)}%

\makeatletter
\newcommand\myisodate{\number\year-\ifcase\month\or 01\or 02\or 03\or 04\or 05\or 06\or 07\or 08\or 09\or 10\or 11\or 12\fi-\ifcase\day\or 01\or 02\or 03\or 04\or 05\or 06\or 07\or 08\or 09\or 10\or 11\or 12\or 13\or 14\or 15\or 16\or 17\or 18\or 19\or 20\or 21\or 22\or 23\or 24\or 25\or 26\or 27\or 28\or 29\or 30\or 31\fi}%
\makeatother
\newcommand*{\abstractnoindent}{}%
\let\abstractnoindent\abstract
\renewcommand*{\abstract}{\let\quotation\quote\let\endquotation\endquote
  \abstractnoindent}
\deffootnote[1em]{1em}{1em}{\textsuperscript{\thefootnotemark}}%
\lstdefinestyle{input}{
  backgroundcolor=\color{semilightgray},%
  commentstyle=\itshape\color{chocolate},%
  keywordstyle=\color{blue},%
  stringstyle=\color{blue},%
  numbers=left,%
  numbersep=4.8pt,%
  numberstyle=\color{darkgray!80}\tiny%
}
\lstdefinestyle{output}{
  backgroundcolor=\color{lightgray}%
}

\lstdefinestyle{Lstyle}{
  language=[LaTeX]TeX,%
  texcs={},%
  otherkeywords={}%
}

\lstnewenvironment{Linput}[1][]{%
  \lstset{style=input, style=Lstyle}
  #1%
}{\vspace{-0.25\baselineskip}}%

\lstnewenvironment{Loutput}[1][]{%
  \lstset{style=output, style=Lstyle}
  #1%
}{\vspace{-0.25\baselineskip}}%

\lstdefinestyle{Rstyle}{
  language=R,%
  keywords={if, else, repeat, while, function, for, in, next, break},%
  otherkeywords={}%
}

\lstnewenvironment{Rinput}[1][]{%
  \lstset{style=input, style=Rstyle}
  #1%
}{\vspace{-0.25\baselineskip}}%

\lstnewenvironment{Routput}[1][]{%
  \lstset{style=output, style=Rstyle}
  #1%
}{\vspace{-0.25\baselineskip}}%

\DeclareNameAlias{sortname}{last-first}%
\DefineBibliographyExtras{american}{\DeclareQuotePunctuation{}}%
\renewbibmacro*{volume+number+eid}{%
  \setunit*{\addcomma\space}%
  \printfield{volume}%
  \printfield{number}}
\DeclareFieldFormat*{number}{(#1)}
\DeclareFieldFormat*{title}{#1}%
\DeclareFieldFormat{doi}{%
  \ifhyperref
    {\href{http://dx.doi.org/#1}{\nolinkurl{doi:#1}}}%
    {\nolinkurl{doi:#1}}}%
\renewbibmacro*{in:}{}%
\DeclareFieldFormat{isbn}{ISBN #1}%
\DeclareFieldFormat{pages}{#1}%
\DeclareFieldFormat{url}{\url{#1}}%
\DeclareFieldFormat{urldate}{\mkbibparens{#1}}%
\renewcommand*{\cite}[2][]{\textcite[#1]{#2}}%

\newif\ifstarttheorem
\newtheoremstyle{mythmstyle}%
{0.5em}%
{0.5em}%
{}%
{}%
{\sffamily\bfseries\global\starttheoremtrue}%
{}%
{\newline}%
{\thmname{#1}\ \thmnumber{#2}\ \thmnote{(#3)}}%
\theoremstyle{mythmstyle}%
\newtheorem{definition}{Definition}[section]%
\newtheorem{proposition}[definition]{Proposition}
\newtheorem{lemma}[definition]{Lemma}

\newtheorem{corollary}[definition]{Corollary}
\newtheorem{remark}[definition]{Remark}
\newtheorem{example}[definition]{Example}
\newtheorem{algorithm}[definition]{Algorithm}

\makeatletter
\preto\itemize{%
  \if@inlabel
    \ifstarttheorem
      \mbox{}\par\nobreak\vskip\glueexpr-\parskip-\baselineskip+0.25em\relax\hrule\@height\z@
    \fi%
  \fi%
  \global\starttheoremfalse%
 \def\tempa{proof}%
 \ifx\tempa\mycurrenvir
    \ifstarttheorem
      \mbox{}\par\nobreak\vskip\glueexpr-\parskip-\baselineskip+0.25em\relax\hrule\@height\z@
    \fi%
 \fi%
 \global\starttheoremfalse%
}
\preto\enditemize{\global\starttheoremfalse}
\makeatother

\makeatletter
\preto\enumerate{%
  \if@inlabel
    \ifstarttheorem
      \mbox{}\par\nobreak\vskip\glueexpr-\parskip-\baselineskip+0.25em\relax\hrule\@height\z@
    \fi%
  \fi%
  \global\starttheoremfalse%
 \def\tempa{proof}%
 \ifx\tempa\mycurrenvir
    \ifstarttheorem
      \mbox{}\par\nobreak\vskip\glueexpr-\parskip-\baselineskip+0.25em\relax\hrule\@height\z@
    \fi%
 \fi%
 \global\starttheoremfalse%
}
\preto\endenumerate{\global\starttheoremfalse}
\makeatother

\newcommand{\ou}[3]{%
  \mathrel{%
    \vcenter{\offinterlineskip
      \ialign{##\cr$#1$\cr\noalign{\kern-#3}$#2$\cr}%
    }%
  }%
}
\newcommand*{\omu}[3]{\underset{#3}{\overset{#1}{#2}}}
\newcommand*{\T}{^{\top}}

\newcommand*{\isim}{\omu{\text{\tiny{ind.}}}{\sim}{}}

\newcommand*{\IR}{\mathbbm{R}}

\newcommand*{\Exp}{\operatorname{Exp}}
\newcommand*{\Poi}{\operatorname{Poi}}

\newcommand*{\U}{\operatorname{U}}

\newcommand*{\N}{\operatorname{N}}

\newcommand*{\I}{\mathbbm{1}}
\newcommand*{\rd}{\mathrm{d}}

\renewcommand*{\P}{\mathbbm{P}}
\newcommand*{\E}{\mathbbm{E}}

\newcommand*{\Var}{\operatorname{Var}}
\newcommand*{\Cov}{\operatorname{Cov}}

\newcommand*{\psii}{{\psi^{-1}}}

\newcommand*{\R}{\textsf{R}}

\begin{document}
%
%
%
%
%
%
%
% \setvruler[10pt][1][1][4][1][0pt][0pt][-30pt][\textheight]%
\thispagestyle{plain}
\begin{center}
  \sffamily
  {\bfseries\LARGE A framework for measuring dependence between random
    vectors\par}
  \bigskip\smallskip
  {\Large Marius Hofert\footnote{Department of Statistics and Actuarial Science, University of
    Waterloo, 200 University Avenue West, Waterloo, ON, N2L
    3G1,
    \href{mailto:marius.hofert@uwaterloo.ca}{\nolinkurl{marius.hofert@uwaterloo.ca}}. The
    author would like to thank NSERC for financial support for this work through Discovery
    Grant RGPIN-5010-2015.},
    Wayne Oldford\footnote{Department of Statistics and Actuarial Science, University of
    Waterloo, 200 University Avenue West, Waterloo, ON, N2L
    3G1,
    \href{mailto:woldford@uwaterloo.ca}{\nolinkurl{rwoldford@uwaterloo.ca}}.},
    Avinash Prasad\footnote{Department of Statistics and Actuarial Science, University of
    Waterloo, 200 University Avenue West, Waterloo, ON, N2L
    3G1,
    \href{mailto:a2prasad@uwaterloo.ca}{\nolinkurl{a2prasad@uwaterloo.ca}}. The author would like to thank NSERC for financial support for this work through a PGS D Scholarship.},
    Mu Zhu\footnote{Department of Statistics and Actuarial Science, University of
    Waterloo, 200 University Avenue West, Waterloo, ON, N2L
    3G1,
    \href{mailto:mu.zhu@uwaterloo.ca}{\nolinkurl{mu.zhu@uwaterloo.ca}}. The
    author would like to thank NSERC for financial support for this work through Discovery Grant RGPIN-2016-03876.}
    \par\bigskip
    \myisodate\par}
\end{center}
\par\smallskip
\begin{abstract}
  A framework for quantifying dependence between random vectors is
  introduced. With the notion of a collapsing function, random vectors are
  summarized by single random variables, called collapsed random variables in
  the framework. Using this framework, a general graphical assessment of
  independence between groups of random variables for arbitrary collapsing
  functions is provided. Measures of association computed from the collapsed
  random variables are then used to measure the dependence between random
  vectors. To this end, suitable collapsing functions are
  presented. Furthermore, the notion of a collapsed distribution function and
  collapsed copula are introduced and investigated for certain collapsing
  functions. This investigation yields a multivariate extension of the Kendall
  distribution and its corresponding Kendall copula for which some properties
  and examples are provided. In addition, non-parametric estimators for the
  collapsed measures of dependence are provided along with their corresponding
  asymptotic properties. Finally, data applications to bioinformatics and
  finance are presented.
\end{abstract}
\minisec{Keywords}
Dependence between random vectors, hierarchical models, collapsing
functions, collapsed random variables, multivariate Kendall distribution and
Kendall copula, graphical test of independence.
\minisec{MSC2010}
62H99, 65C60, 62-09%

\section{Introduction}
While there are numerous well established methods to measure dependence between
random variables, the extension to random vectors (for example, for modeling groups of random variables)
poses a significant challenge. This challenge arises from the lack of a unique
axiomatic framework that outlines desirable properties a measure of dependence
between random vectors should exhibit. Moreover, there is no unique
extension of bivariate measures of association to arbitrary dimensions and the
multivariate measures of association available do not naturally capture
dependence between more than one random vector as is of interest in applications
such as bioinformatics, finance, insurance or risk management.

Proposed solutions to this problem are rather difficult to find in the
literature. A classical methodology to summarize (linear) dependence between
random vectors is the well known canonical correlation coefficient; see
\cite{hotelling1936}. A non-linear extension of canonical correlation has been
suggested through the use of kernel functions in \cite{bach2002} and
\cite{ghoraie2015}. A faster version of the kernel canonical correlation method
has been developed by adopting the idea of randomized kernels; see
\cite{lopez2013}. \cite{szekelyrizzobakirov2007} proposed a novel distance
covariance coefficient, defined as a weighted $\mathcal{L}^2$-norm between the
joint characteristic function and product of marginal characteristic functions
of the random vectors under consideration.
In the context of copula modeling,
\cite{grothe2014} recently derived versions of Spearman's rho and Kendall's tau
between random vectors and corresponding estimation
procedures. Our framework will generalize their approach. This generalization
will also allow us to derive a couple of interesting results as by-products of
our framework.

Note that there is neither an inherently correct nor a
canonical %
way of measuring dependence between random vectors. As a result, one
can think of multiple variations of quantifying such dependence. Approaches are
primarily motivated by the purpose, for example, detection or ranking of
dependencies and the salient features of the dataset under investigation. In
this paper, we subsume several such approaches under a general framework which
allows us to detect, quantify, visualize and check dependence between
random vectors.

The paper is organized as follows. In Section~\ref{sec:framework} we present the
said framework for measuring dependence between random vectors and utilize it to
develop a visual assessment of independence between random vectors. Section
~\ref{sec:joint:Kendall} develops the notion of a collapsed distribution
function and collapsed copula with explicit characterizations for some choices
of collapsing functions (to be detailed later). In Section~\ref{sec:fit:S} we
discuss non-parametric estimators for the dependence measures introduced in
Section~\ref{sec:framework} and their corresponding asymptotic
properties. Empirical examples from the realm of bioinformatics and finance are
covered in Section~\ref{sec:ex}. Section 6 concludes.

\section{The framework}\label{sec:framework}

For introducing a framework for measuring dependence between random vectors, it
suffices to consider the case of two, a $p$-dimensional random vector
$\bm{X}=(X_1,\dots,X_p)$ (with continuous marginal distribution functions
$F_{X_1},\dots,F_{X_p}$) and a $q$-dimensional random vector
$\bm{Y}=(Y_1,\dots,Y_q)$ (with continuous marginal distribution functions
$F_{Y_1},\dots,F_{Y_q}$), defined on some probability space with probability
measure $\P$. Our target is to measure dependence between $\bm{X}$ and $\bm{Y}$
with a measure of association
\begin{align*}
  \chi=\chi(\bm{X},\bm{Y})
\end{align*}
mapping to $[-1,1]$; note the missing ``the'' before ``dependence'',
depending on the context, various notions of dependence are possible. As
mentioned before, this is different from multivariate extensions of measures of
association which aim to summarize dependence of a single random vector
(say, just $\bm{X}$); see \cite{schmid2010} and the references therein for a
comprehensive treatment.

A natural first step is to establish properties $\chi$ should
satisfy. For bivariate measures of association, that is, measures of association
between two random variables $X$ and $Y$, such properties were listed in
\cite{renyi1959} and with minor revisions later in \cite{schweizerwolff1981}.
\cite{scarsini1984} introduced the notion of concordance measures by adding the
property that measures should respect a pointwise partial ordering on the set of
copulas also known as concordance ordering; see
\cite{embrechtsmcneilstraumann2002} for more on concordance (or
rank-correlation) measures and their motivation due to pitfalls of (linear)
correlation. Prominent examples are Kendall's tau and Spearman's rho. Another
type of bivariate measure of association, focusing on the (extremal) dependence
in the joint tails of a bivariate distribution, is the (lower or upper)
coefficient of tail dependence.

Note that, more recently, \cite{reshef2011} described an ideal measure of association in
the bivariate case as so-called equitable dependence measures. The notion of
an equitable dependence measure extends the invariance property of concordance
measures to include invariance under non-monotone marginal transforms. However, the
maximal information coefficient (MIC) introduced in \cite{reshef2011}, which
supposedly satisfies this equitability condition, is purely data-driven and
heuristic. As a result, the MIC measure does not naturally fit into our
probabilistic framework. Various versions of this equitability condition have
since been proposed including more mathematically formal definitions; see, for
example, \cite{kinney2014}. Hence, there is some consensus of an ``ideal''
bivariate measure of association but our problem demands generalizations of
these properties to vector-based measures of association, which is non-trivial.

\cite{grothe2014} recently approached this problem and listed properties of a
concordance measure which easily carry over from random variables $X,Y$ to random
vectors $\bm{X},\bm{Y}$. These include:
\begin{enumerate}
\item $\chi(\bm{X},\bm{Y})\in[-1,1]$;
\item $\chi(\bm{X},\bm{Y})$ is invariant to permutations of the components of the
  random vectors $\bm{X}$ and $\bm{Y}$;
\item independence of $\bm{X}$ and $\bm{Y}$ implies $\chi(\bm{X},\bm{Y})=0$.
\end{enumerate}
The translation of more non-trivial properties such as the invariance to (some
sort of) increasing transformations of $\bm{X}$ and $\bm{Y}$ and the concordance ordering to
the vector case is less transparent. One such generalization  of these properties proposed in \cite{grothe2014} is as follows.
\begin{enumerate}
\item \emph{Invariance property}. $\chi(\bm{X},\bm{Y})$ is invariant to
  increasing transformations of the components of the random vectors
  $\bm{X}$ and $\bm{Y}$;
\item \emph{Concordance ordering property}. Suppose one has two distribution
  functions with margins $F_{X_1},\dots,F_{X_p},F_{Y_1},\dots,F_{Y_q}$ and
  copulas $C_1$ and $C_2$, respectively, such that $C_1\preceq C_2$, that is,
  $C_1(\bm{u})\le C_2(\bm{u})$ for all $\bm{u}\in[0,1]^{p+q}$. Then
  $\chi_{C_1}(\bm{X},\bm{Y}) \leq \chi_{C_2}(\bm{X},\bm{Y})$, where
  $\chi_{C_1}$ and $\chi_{C_2}$ denote the measures of association expressed as
  functions of only the copula $C_1$ and $C_2$, respectively.
\end{enumerate}
The difficulty lies in hypothesizing invariance and concordance properties when
the marginal distributions $F_{X_1},\dots,F_{X_p},F_{Y_1},\dots,F_{Y_q}$ and the
copula $C_{\bm X}, C_{\bm Y}$ of $\bm X$, $\bm Y$, respectively, all can vary;
generalizing the concept of equitable dependence faces similar
difficulties.

\subsection{The collapsed random variables}
The framework we suggest consists of collapsing or summarizing the two random
vectors $\bm{X}$ and $\bm{Y}$ to single random variables $S(\bm{X})$ and
$S(\bm{Y})$ referred to as \emph{collapsed random variables}. The function $S$
maps random vectors to random variables and is referred to as a
\emph{collapsing} (or \emph{summary}) \emph{function}. Note that in the
  most general setup, we could use separate collapsing functions,
  $S_{\bm{X}}(\bm{X})$ and $S_{\bm{Y}}(\bm{Y})$. Different collapsing
  functions could be particularly useful when $\bm{X}$ and $\bm{Y}$ lie in
  different domains, that is, continuous vs discrete. However, in what follows, we will
  for the sake of simplicity restrict ourselves to using the same collapsing function $S$ to
  collapse $\bm{X}$ and $\bm{Y}$, and remain in the continuous domain to
  facilitate development of theoretical results. The (bivariate) distribution
function of $(S(\bm{X}),S(\bm{Y}))$ is called the \emph{collapsed distribution
  function} in our framework and its copula (if unique) is termed \emph{collapsed copula}; see
Section~\ref{sec:joint:Kendall} for more details.

Collapsing functions for different random vectors typically are of similar
functional form; see Section~\ref{sec:choose:S} for several examples. However,
they can differ, for example, due to the different dimensions of $\bm{X}$ and
$\bm{Y}$. Furthermore, as we will see in Section~\ref{sec:choose:S}, a
collapsing function for $\bm{X}$ does not necessarily have to be a $p$-variate
function, it can also be a $2p$-variate function (denoted as $S(\bm{X},\bm{X}')$,
where $\bm{X}'$ is an independent copies of $\bm{X}$), for example. As such, the
notion of a collapsing function is quite general, the only requirement being
that a random vector is mapped to a single random variable. For simplicity, we
denote all collapsing functions by $S$ and speak of the collapsing function as
being applied to $\bm{X}$ or to $\bm{Y}$. Options for $S$ are provided in
Section~\ref{sec:choose:S} and estimation is addressed in
Section~\ref{sec:fit:S}; concrete choices of $S$ are also provided and
discussed, for example, in the applications in Section~\ref{sec:ex}. We start by
presenting a general graphical assessment (can be converted to a statistical test if needed) of independence between groups of random variables.

\subsection{A graphical assessment of independence between random vectors}

As mentioned in the introduction, \cite{szekelyrizzobakirov2007}
suggest a formal test of independence between $\bm X$ and $\bm Y$ based on the
distance between the characteristic function of $(\bm X,\bm Y)$ and the product
of the characteristic functions of $\bm X$ and $\bm Y$. Furthermore, they
derived that this distance can equivalently be expressed as a function of the
correlation coefficient of Euclidean distances. Using this formal test as
motivation, \cite{wang2013} introduced a graphical test of independence between
random variables with Euclidean distance and rank transform. We further extend
this work to a graphical assessment of independence between groups of random
variables with various different transforms (that is, collapsing functions).

The method we suggest is based on the Grouping Lemma, see
\cite[Lemma~4.4.1]{resnick2014}, which states that measurable functions of
independent random variables are independent; see also
\cite[Theorem~2.1.6]{durrett2004}. This result can be conveniently used to
construct a test of independence between two or more groups of random variables
by testing independence of the collapsed random variables in our framework; note
that the corresponding hypothesis $\mathcal{H}_{0,c}$ tested (namely the
collapsed random variables to be independent) is only a subset of the hypothesis
$\mathcal{H}_0$ that all random variables are independent.

In principle, all known statistical tests of independence between two
or more random variables can be applied for testing
$\mathcal{H}_{0,c}$. What we suggest here is a graphical assessment
for $\mathcal{H}_{0,c}$. As is typically of interest in practice, see also our
example in Section~\ref{sec:ex:SP500}, we consider $G\ge 2$ random vectors here.
\begin{algorithm}[Graphical assessment of independence of groups of variables]\label{vis:test:indep}
	Let $\bm{X}_{i1},\dots,\bm{X}_{iG}$, $i\in\{1,\dots,n\}$, be a random sample
	from $G$ groups of random variables $\bm{X}_1,\dots,\bm{X}_G$ of dimensions $p_1,\dots,p_G$.
	To visually check independence of the groups of random variables
	$\bm{X}_1,\dots,\bm{X}_G$ based on the sample $\bm{X}_{i1},\dots,\bm{X}_{iG}$,
	$i\in\{1,\dots,n\}$, do:
	\begin{enumerate}
		\item For each group $g\in\{1,\dots,G\}$ of variables, compute the collapsed
		variables $S_{ig}=S(\bm{X}_{ig})$, $i\in\{1,\dots,k\}$, where $k=n$ for $p$-variate functions and $k=\binom{n}{2}$ for $2p$-variate functions.
		\item Compute the pseudo-observations
		\begin{align*}
		U_{k,ig}=\frac{R_{ig}}{k+1},\quad i\in\{1,\dots,k\},\ g\in\{1,\dots,G\},
		\end{align*}
		where, for each $g\in\{1,\dots,G\}$, $R_{ig}$ denotes the rank of $S_{ig}$
		among $S_{1g},\dots,S_{kg}$.
              \item Visualize all pairs of pseudo-observations
                $(U_{k,ig},U_{k,ih})$, $i\in\{1,\dots,k\}$,
                $g,h\in\{1,\dots,G\}:g<h$. This can be done in a scatter-plot
                matrix (for small to moderate $G$) or with a zenplot (for large
                $G$); see \cite{hofertoldford2017} and Section~\ref{sec:ex}
                for the latter. The less the visualized
                samples resemble realizations from $\U(0,1)^2$ the greater the
                evidence against $\mathcal{H}_{0,c}$ and thus $\mathcal{H}_0$.
	\end{enumerate}
\end{algorithm}

Note that we can turn this graphical assessment into a statistical test of
independence by adopting the line up test proposed in \cite{buja2009}. That is,
if in addition to the visualized samples, we also displayed groups of
independent realizations from $\U(0,1)^2$, then actual significance levels could
be determined.

An interesting question is whether our visual assessment of independence is
independent of the marginal distributions of any of the $d=p_1+\dots+p_G$
components of $(\bm X_1,\dots,\bm X_G)$. This certainly depends on the
collapsing function. In general, it does not matter for an assessment of
independence, but for better interpretability (of the visualized pseudo-observations) one could of course build
pseudo-observations of the given data from $(\bm X_1,\dots,\bm X_G)$ before
applying Algorithm~\ref{vis:test:indep}; note that in this case, one would apply
pseudo-observations at two levels, to the original variables and the collapsed
variables.

The distance correlation test developed in \cite{szekelyrizzobakirov2007} is a
notable statistical test of independence between random vectors. In particular,
the distance correlation (population) test statistic possesses the desirable
property that it is zero if and only if $\bm{X}$ and $\bm{Y}$ are independent,
thus making it particularly useful for testing independence. With an
appropriately chosen collapsing function (see Table~\ref{tab:collapsing} for
examples), our framework could yield a more powerful (graphical) test of
independence. The main advantages when compared with distance correlation are
that we are working with pseudo-observations %
and that there are many different types of departures from independence that can
be observed in comparison to a single numerical test statistic.

\subsection{Collapsed measures of association and dependence}
After $\bm{X}$ and $\bm{Y}$ have been collapsed to $S(\bm{X})$ and $S(\bm{Y})$,
respectively, the latter two random variables can be used to detect, quantify and
check dependence between $\bm{X}$ and $\bm{Y}$ using a classical and
well understood bivariate measure of association referred to as \emph{collapsed
  measure of association} in our framework. Although there are various choices
of collapsed measures of association (including, for example, tail dependence;
see also later), we will mainly focus on Pearson's correlation coefficient $\rho$ and
thus consider
\begin{align}
  \chi(\bm{X},\bm{Y})=\rho(S(\bm{X}),S(\bm{Y}))\label{eq:measure}
\end{align}
as a measure of association between $\bm{X}$ and $\bm{Y}$. The choice of
Pearson's correlation coefficient seems careless given the known deficiencies of
correlation for quantifying monotone dependence or concordance (as opposed to
``just'' linear dependence); see, for example, \cite{embrechtsmcneilstraumann2002}. However,
Spearman's rho and Kendall's tau both appear as a special case of \eqref{eq:measure} when choosing appropriate collapsing
functions $S$. This is obvious in the case of Spearman's rho $\rho_{\text{S}}$ which
is simply Pearson's correlation $\rho$ of the (univariate) probability integral
transformed random variables. In other words, if
$F_{S(\bm{X})}$ denotes the daistribution function of the collapsed random
variable $S(\bm{X})$, we can use
the collapsing functions $\tilde{S}(\bm{x})=F_{S(\bm{X})}(S(\bm{x}))$ and $\tilde{S}(\bm{y})=F_{S(\bm{Y})}(S(\bm{y}))$ to obtain
\begin{align*}
  \chi(\bm{X},\bm{Y})=\rho(\tilde{S}(\bm{X}),\tilde{S}(\bm{Y}))=\rho_{\text{S}}(S(\bm{X}),S(\bm{Y})).
\end{align*}
The following lemma shows that also Kendall's tau appears as a special case of
\eqref{eq:measure}; note that the collapsing function
is an example of a $2p$-variate collapsing function as mentioned before.
\begin{lemma}[Kendall's tau as a special case of
  \eqref{eq:measure}]\label{lem:repr:tau:cor}
  Let $\bm X$ and $\bm Y$ be continuously distributed random vectors and let
  $\bm X'$ and $\bm Y'$ be independent copies of $\bm X$ and $\bm Y$,
  respectively. Under our framework, the collapsing function
  $\tilde{S}(\bm{X},\bm{X}')=\I_{\{S(\bm{X})\le S(\bm{X}')\}}$ leads to
  \begin{align*}
    \chi(\bm{X},\bm{Y})=\rho(\tilde{S}(\bm{X},\bm{X}'),\tilde{S}(\bm{Y},\bm{Y}'))=\tau(S(\bm{X}),S(\bm{Y})). %\label{eq:tau:cor}
  \end{align*}
\end{lemma}
\begin{proof}
  Let $S_1,S_2$ be continuously distributed random variables and let $S_1',S_2'$ be
  independent copies of $S_1,S_2$, respectively. Then
  \begin{align*}
    \tau (S_1,S_2)&=\P((S_2-S_2')(S_1-S_1')>0) - \P((S_2-S_2')(S_1-S_1')<0)\\
         &=2\P((S_2-S_2')(S_1-S_1')>0) - 1=4 \P\bigl(S_1 \le S_1',\ S_2 \le
           S_2'\bigr) - 1\\
         &= \frac{\E(\I_{\{S_1\le S_1',\ S_2\le
           S_2'\}})-\frac{1}{2}\cdot\frac{1}{2}}{\sqrt{\frac{1}{2}\Bigl(1-\frac{1}{2}\Bigr)\frac{1}{2}\Bigl(1-\frac{1}{2}\Bigr)}}=\frac{\E(\I_{\{S_1\le
           S_1'\}}\I_{\{S_2\le S_2'\}})-\E(\I_{\{S_1\le
           S_1'\}})\E(\I_{\{S_2\le S_2'\}})}{\sqrt{\Var(\I_{\{S_1\le
           S_1'\}})\Var(\I_{\{S_2\le S_2'\}})}}\\
         &=\rho(\I_{\{S_1\le S_1'\}}, \I_{\{S_2\le S_2'\}}),
  \end{align*}
  that is, Kendall's tau equals the correlation coefficient of the indicators
  $\I_{\{S_1\le S_1'\}}$ and $\I_{\{S_2\le S_2'\}}$. With the collapsing function
  $\tilde{S}$ as claimed (and $S_1=S(\bm{X})$, $S_1'=S(\bm{X}')$, $S_2=S(\bm{Y})$, $S_2'=S(\bm{Y}')$), we thus obtain that
  \begin{align*}
    \chi(\bm{X},\bm{Y})=\rho(\tilde{S}(\bm{X},\bm{X}'),\tilde{S}(\bm{Y},\bm{Y}'))=
                         \rho(\I_{\{S(\bm X)\le S(\bm X')\}}, \I_{\{S(\bm Y)\le S(\bm
                         Y')\}})=\tau(S(\bm X), S(\bm Y)).
  \end{align*}
\end{proof}

Finally let us briefly address the notion of tail dependence; see, for example,
\cite[Section~2.1.10]{joe1997} or \cite[Section~5.4]{nelsen2006}. Although there
are multivariate notions of tail dependence, see
\cite[Chapter~10]{jaworskidurantehaerdlerychlik2010} for an overview, there is,
to the best of our knowledge, no notion of tail dependence between two random
vectors $\bm{X},\bm{Y}$. An intuitive choice in our framework is simply
\begin{align}
\chi(\bm{X},\bm{Y})=\lambda(S(\bm{X}),S(\bm{Y}))\label{chi:lam}
\end{align}
where $\lambda$ denotes the lower or upper coefficient of tail dependence as
implied by the collapsed copula. This concept can be extended to more than two
random vectors by considering matrices; see \cite{embrechtshofertwang2016}.

Let us now go back to the case where the collapsed measure of association is
Pearson's correlation coefficient.

\subsection{Choosing the collapsing function $S$}\label{sec:choose:S}
The choice of the collapsing function $S$ for measuring dependence between
random vectors is fairly open ended. We start by introducing various
options for $S$, summarized in Table~\ref{tab:collapsing}. As before,
$\bm{X}'$ is used to denote an independent copy of the random vector $\bm{X}$.
\begin{table}[htbp]
  \centering
  \begin{tabular}{ll}
    \toprule
    \multicolumn{1}{c}{Type of $S$}&\multicolumn{1}{c}{Collapsing function $S$}\\
    \midrule
    Weighted average&$S(\bm{X})=\bm{w}\T\bm{X}\ \text{for}\ \bm{w}=(w_1,\dots,w_p),\ \sum_{j=1}^{p}w_j=1$\\
    Maximum (or minimum)&$S(\bm{X})=\max\limits_{1\le j\le p}\{X_j\}$
                          (or $S(\bm{X})=\min\limits_{1\le j\le p}\{X_j\}$)\\
    Distance&$S(\bm{X},\bm{X}')=D(\bm{X},\bm{X}')$\\
    Kernel similarity&$S(\bm{X},\bm{X}')=K(\bm{X},\bm{X}')$\\
    Multivariate rank&$S(\bm{X},\bm{X}')=\I_{\{\bm{X}\le \bm{X}'\}}$\\
    Probability integral transform&$S(\bm{X})=F_{\bm{X}}(\bm{X})$\\
    \bottomrule
  \end{tabular}
  \caption{Examples of collapsing functions $S$ of a random vector $\bm{X}$ (with independent copy $\bm{X}'$); note that the inequality $\bm{X}\leq \bm{X}'$ in the multivariate rank collapsing function is understood componentwise.}
  \label{tab:collapsing}
\end{table}
Note that if $F_{X_1},\dots,F_{X_p}$ are continuous, the multivariate rank transform satisfies
\begin{align*}
  S((X_1,\dots,X_p),\ (X_1',\dots,X_p'))&=\I_{\{X_1\le X_1',\dots,X_p\le
                                          X_p'\}}\\
  &=\I_{\{F_{X_1}(X_1)\le F_{X_1}(X_1'),\dots,F_{X_p}(X_p)\le F_{X_p}(X_p')\}}\\
  &=S\bigl((F_{X_1}(X_1),\dots,F_{X_p}(X_p)),\ (F_{X_1}(X_1'),\dots,F_{X_p}(X_p'))\bigr)
\end{align*}
and thus does not depend on the specific marginal distributions $F_{X_1},\dots,F_{X_p}$
involved; furthermore, $F_{X_j}(X_j)\sim\U(0,1)$, $j\in\{1,\dots,p\}$. Similarly for
the probability integral transform (PIT), by Sklar's Theorem, see \cite{sklar1959},
\begin{align*}
  S(X_1,\dots,X_p)&=F_{(X_1,\dots,X_p)}(X_1,\dots,X_p)=C_{\bm X}(F_{X_1}(X_1),\dots,F_{X_p}(X_p))\\
  &=S(F_{X_1}(X_1),\dots,F_{X_p}(X_p)),
\end{align*}
Therefore, the PIT as collapsing function also does not depend on the specific
marginal distributions involved.

For collapsing functions which are not invariant under the marginal
distributions (such as weighted average, maximum, minimum, distance and kernel
similarity), one can easily introduce such property by replacing
$(X_1,\dots,X_p)$ by $(F_{X_1}(X_1),\dots,F_{X_p}(X_p))$ in $S$. This is often
useful for getting a (rank-based) picture of dependence independently of the
margins which can be of interest for visualization or estimation
purposes (empirically, this means computing pseudo-observations).

The following sections consider each of the collapsing functions listed in Table~\ref{tab:collapsing} in more detail.

\subsubsection{The weighted average collapsing function}
The weighted average function is a classical choice of collapsing function.
Here are a few ways how the weights $w_1,\dots,w_p$ can be chosen:
\begin{enumerate}
\item \textbf{Equal weights}. For equal weights $w_j=\frac{1}{p}$,
  $j\in\{1,\dots,p\}$, we obtain the simple average as collapsing function. %
\item \textbf{Application-specific weights}. One can choose weights $w_j$
  which are tailor-made for a specific application in mind. See
  Section~\ref{sec:ex} for examples.
\item \textbf{Dimension reduction weights}. Typically by adopting any dimension
  reduction technique, one could use some normalized version of the measure
  (usually singular values of a specific matrix depending on the technique) used
  to order the dimensions as weights in our framework.
\item \textbf{Optimal weights}. One can choose optimal weights with respect
  to some objective function. For example, analogously to the notion of canonical correlation, one could empirically choose the weights for every pair of random vectors such that the resulting measure of association $\chi$ is maximized. For example, if $\bm{X},\bm{Y}$ are elliptical, one could consider
  \begin{align*}
    \chi(\bm{X},\bm{Y})=\sup_{\bm{w}_1\in\IR^p \, \bm{w}_2 \in \IR^q} \rho(\bm{w}_1\T\bm{X},\bm{w}_2\T\bm{Y})
   \end{align*}
\item \textbf{Extreme weights}. One can consider the \emph{$m$-largest} (or \emph{$m$-smallest})
  \emph{weighted average}, that is, the average over the $m$ largest (or $m$ smallest) order statistics
  per group of random variables. This could be of interest in the
  context of financial risk management, where one needs to keep track of the $m$ largest (or $m$ smallest)
  losses in two or more portfolios or asset classes.
\end{enumerate}

\subsubsection{The maximum collapsing function}
The componentwise maximum (or minimum) is a special case of the aforementioned
extreme weighted case, with 1-largest (or 1-smallest) weighted average as collapsing function, that is,
\begin{align*}
  S(\bm{X})=\max\{X_1,\dots,X_p\}\quad\text{(or $S(\bm{X})=\min\{X_1,\dots,X_p\}$)}.
\end{align*}
This requires all dimensions of the random vector $\bm{X}$ to have a comparable
interpretation and makes sense, for example, for quantifying dependence between
market return data grouped into sectors. In this case we would be measuring dependence between different market sectors through the best (or worst) performer in each sector.

\subsubsection{The pairwise distance collapsing function}
The population version of the pairwise distance collapsing function requires
invoking an independent copy of the random vector. For the sample version based
on sample size $n$, this implies that distances are computed between the
$\binom{n}{2}$ distinct pairs of the $n$ samples. One can choose virtually any type of distance $D$.  For example, some
of the standard distance functions we experimented with in Section~\ref{sec:ex}
are Euclidean, Manhattan, Canberra, and Minkowski. Some distance functions for
non-continuous measurements include cosine distance (suitable for text data),
hamming distance (datasets in information theory), and Jaccard distance.

Ideally one should have a data- or application-specific reason to choose distances other than Euclidean distance (but numerical experiments have shown that it can
sometimes be advantageous to choose the Canberra distances to avoid
issues due to large distances).

\subsubsection{The pairwise kernel collapsing function}
Similar to the distance transform, the kernel collapsing function $K$ results in
$\binom{n}{2}$ samples in the transformed space. As for the choice of $K$, one
can choose any kernel function, some of which are listed in Table~\ref{tab:kernel}.
\begin{table}[htbp]
  \centering
  \begin{tabular}{ll}
    \toprule
    \multicolumn{1}{c}{Type of $K$}&\multicolumn{1}{c}{Kernel function $K(\cdot\,;\cdot):\IR^p\times\IR^p\rightarrow\IR$}\\
    \midrule
    Linear (trivial)&$K(\bm{x}_i,\bm{x}_k)=\bm{x}_i\T\bm{x}_k$\\
    Polynomial (of order $d$)&$K(\bm{x}_i,\bm{x}_k)=(1+\bm{x}_i\T\bm{x}_k)^d$\\
    Gaussian &$K(\bm{x}_i,\bm{x}_k)=\exp\bigl(-(\|\bm{x}_i-\bm{x}_k\|_2^2/(2\sigma^2))\bigr)$\\
    von Mises &$K(\bm{x}_i,\bm{x}_k)= \prod_{t=1}^{p} \exp(\kappa_t\cos(x_{it}-x_{kt})) $\\
    \bottomrule
  \end{tabular}
  \caption{Examples of kernel functions.}
  \label{tab:kernel}
\end{table}
By default one can choose the Gaussian kernel unless the
peculiarities of a dataset or application context suggest a potential
alternative like the von Mises kernel which was used in Section~\ref{sec:ex:prot} for the protein dataset.

\subsubsection{The multivariate rank collapsing function} \label{subsec:multirank}
Utilizing the multivariate rank collapsing function $S(\bm{X},\bm{X}')=\I_{\{\bm{X}\le\bm{X}'\}}$ to summarize multidimensional
random vectors to a single dimension yields a rank-based measure of association
$\chi$; as usual, the inequality $\bm{X}\le\bm{X}'$, is understood
componentwise. The resulting association measure was first introduced in
\cite{grothe2014} as one possible multivariate extension of Kendall's tau. As is
evident, this particular multivariate Kendall's tau naturally fits into our
framework with the aforementioned choice of collapsing function. Rank-based measures of association possess certain attractive properties,
including the invariance property and the concordance ordering
property as outlined at the onset of this section. Furthermore, as argued in
\cite{grothe2014}, this particular dependence measure $\chi$ can effectively
detect negative association between random vectors.

\subsubsection{The probability integral transform collapsing function}
The probability integral transform collapsing function bears some resemblance to
the multivariate extension of Spearman's rho discussed in
\cite{grothe2014}. However, the definition according to our framework of the
population version of $\chi$ and hence the estimation procedure differs. The
PIT-transformed collapsed random variable $F_{\bm{X}}(\bm{X})$ has distribution
function $K_{\bm{X}}(t)=\P(F_{\bm{X}}(\bm{X})\le t)$, $t\in[0,1]$, known as
\emph{Kendall distribution}. Since
$F_{\bm{X}}(\bm{X})=C_{\bm{X}}(F_{X_1}(X_1),\dots,F_{X_p}(X_p))=C_{\bm{X}}(U_1,\dots,U_p)$
for $\bm{U}=(U_1,\dots,U_p)\sim C_{\bm{X}}$, $K_{\bm{X}}$ only depends on the
copula $C_{\bm{X}}$ of $\bm{X}$ and can thus be viewed as a summary of the
dependence among the components of $\bm X$ in the form of a $p$-variate
function. Unfortunately, $K_{\bm{X}}$ itself is rarely analytically tractable
for dimensions of $\bm{X}$ larger than two, an exception being Archimedean
copulas $C_{\bm{X}}$ with generators $\psi$ for which a calculation based on the
stochastic representation and a connection with the Poisson distribution
function can be used conveniently to show that
\begin{align}
 K_{\bm{X}}(t)=\sum_{k=0}^{p-1}\frac{\psi^{(k)}(\psii(t))}{k!}(-\psii(t))^k,\quad t\in[0,1];\label{eq:univ:K:AC}
\end{align}
see the proof of Proposition~\ref{prop:joint:K:AC} below for this approach or
\cite{barbegenestghoudiremillard1996} for the first appearance of this
result. Working with the multivariate PIT collapsing function and
corresponding Kendall distribution naturally motivates a multivariate
extension of the latter. Various properties and examples associated with such joint Kendall distributions, viewed as an example of a collapsed distribution function, are presented in \ref{sec:PIT:collapsed}.

\section{Collapsed distribution functions and copulas}\label{sec:joint:Kendall}

While we can always compute and visualize realizations from the empirical
collapsed copula (see Algorithm~\ref{vis:test:indep}), deriving an explicit
characterization of the collapsed distribution function or copula in terms of
the joint distribution of $\bm{Z}=(\bm{X},\bm{Y})$ is challenging. To this end,
we present some results for the maximum and PIT collapsing functions. Most notably,
characterizing the collapsed distribution function of the PIT collapsed random
variables, yields a multivariate extension of the Kendall distribution.

\subsection{Maximum collapsing function}
\begin{proposition}[The collapsed distribution and its copula for the maximum
	collapsing function]\label{prop:S:max:cop}
	Let $X_1,\dots,X_p,Y_1,\dots,Y_q$ be continuously distributed random variables
	with distribution functions $F_{X_1},\dots,F_{X_p},F_{Y_1},\dots,F_{Y_q}$,
	respectively.  Furthermore, let $F_{\bm{X},\bm{Y}}$ denote the
	distribution function of $(\bm{X},\bm{Y})$ and consider the maximum collapsing
	function $S$. Then the collapsed distribution function
	$F_{(S(\bm{X}),S(\bm{Y}))}$ is
	$F_{(S(\bm{X}),S(\bm{Y}))}(x,y)=F_{\bm{X},\bm{Y}}(x,\dots,x,y,\dots,y)$ with
	corresponding collapsed copula
	\begin{align*}
	C_{S(\bm{X}),S(\bm{Y})}(u,v)=F_{\bm{X},\bm{Y}}(F_{S(\bm{X})}^-(u),\dots,F_{S(\bm{X})}^-(u),F_{S(\bm{Y})}^-(v),\dots,F_{S(\bm{Y})}^-(v)),\quad u,v\in[0,1],
	\end{align*}
	where $F_{S(\bm{X})}^-(u)$ and $F_{S(\bm{Y})}^-(v)$ denote the quantile functions of the distribution
	functions $F_{S(\bm{X})}(x)=F_{\bm{X},\bm{Y}}(x,\dots,x,\infty,\dots,\infty)$
	and $F_{S(\bm{Y})}(y)=F_{\bm{X},\bm{Y}}(\infty,\dots,\infty,y,\dots,y)$, respectively.
\end{proposition}
\begin{proof}
	Since
	$F_{(S(\bm{X}),S(\bm{Y}))}(x,y)=\P(\max\{X_1,\dots,X_p\}\le x,\
	\max\{Y_1,\dots,Y_q\}\le y)=\P(X_1\le x,\dots, X_p\le x,Y_1\le y,\dots,Y_q\le
	y)=F_{\bm{X},\bm{Y}}(x,\dots,x,y,\dots,y)$ with margins
	$F_{S(\bm{X})}(x)=F_{\bm{X},\bm{Y}}(x,\dots,x,\infty,\dots,\infty)$ and
	$F_{S(\bm{Y})}(y)=F_{\bm{X},\bm{Y}}(\infty,\dots,\infty,y,\dots,y)$, Sklar's
	Theorem implies that the collapsed copula $C_{S(\bm{X}),S(\bm{Y})}$ is given
	as stated.
\end{proof}

Using this setup, we can derive some properties of the collapsed copula to
demonstrate that the maximum collapsing function can intuitively capture
dependence between random vectors. To this end, we call $\bm{X}$ and
  $\bm{Y}$ \emph{comonotone} (\emph{countermonotone}) if
  $\bm{X}=(F_{X_1}^-(U),\dots,F_{X_p}^-(U))$ and
  $\bm{Y}=(F_{Y_1}^-(U),\dots,F_{Y_q}^-(U))$ ($\bm{X}=(F_{X_1}^-(U),\dots,F_{X_p}^-(U))$ and
  $\bm{Y}=(F_{Y_1}^-(1-U),\dots,F_{Y_q}^-(1-U))$) for $U\sim\U(0,1)$; see also Proposition~\ref{prop:prop:Kendall:cop} where this concept is used.
\begin{proposition}[Basic properties of maximum collapsed copulas]
  Let $\bm{X}\sim F_{\bm{X}}$ be a $p$-dimensional and $\bm{Y}\sim F_{\bm{Y}}$
  be a $q$-dimensional random vector, both with continuously distributed
  margins (denoted as before).
  \begin{enumerate}
  \item If $\bm{X}$ and $\bm{Y}$ are independent, then $C_{S(\bm{X}),S(\bm{Y})}(u,v)=uv$ for $u,v \in [0,1]$.
  \item If $\bm{X}$ and $\bm{Y}$ are comonotone and each have equal margins, then
    $C_{S(\bm{X}),S(\bm{Y})}(u,v)=\min\{u,v\}$
    and thus the collapsed copula in this case is the upper Fr\'echet--Hoeffding bound.
  \item If $\bm{X}$ and $\bm{Y}$ are countermonotone and each have equal margins, then
    $C_{S(\bm{X}),S(\bm{Y})}(u,v)=\max\{u+v-1, 0\}$ and thus the
    collapsed copula in this case is the lower Fr\'echet--Hoeffding
    bound.
  \end{enumerate}
  \begin{proof}
    \begin{enumerate}
    \item Let $F_{\bm{X},\bm{Y}}$ denote the distribution function of $(\bm{X},\bm{Y})$.
      Independence and Proposition~\ref{prop:S:max:cop} imply that
      \begin{align*}
        C_{S(\bm{X}),S(\bm{Y})}(u,v)&=F_{\bm{X},\bm{Y}}(F_{S(\bm{X})}^-(u),\dots,F_{S(\bm{X})}^-(u),F_{S(\bm{Y})}^-(v),\dots,F_{S(\bm{Y})}^-(v))\\
                                    &=F_{\bm{X}}(F_{S(\bm{X})}^-(u),\dots,F_{S(\bm{X})}^-(u))\,F_{\bm{Y}}(F_{S(\bm{Y})}^-(v),\dots,F_{S(\bm{X})}^-(v))\\
                                    &=F_{\bm{X},\bm{Y}}(F_{S(\bm{X})}^-(u),\dots,F_{S(\bm{X})}^-(u),\infty,\dots,\infty)\\
                                    &\phantom{{}=}\cdot F_{\bm{X},\bm{Y}}(\infty,\dots,\infty,F_{S(\bm{Y})}^-(v),\dots,F_{S(\bm{X})}^-(v))\\
                                    &=F_{S(\bm{X})}(F_{S(\bm{X})}^-(u))F_{S(\bm{Y})}(F_{S(\bm{Y})}^-(v))=uv.
      \end{align*}
    \item\label{prop:max:coll:cop:2} $F_{S(\bm X)}(x)=\min_{1\le j\le p}\{F_{X_j}(x)\}=F_{X}(x)$ so that $F_{S(\bm X)}^-(u)=\max_{1\le j\le p}\{F_{X_j}^-(u)\}=F_X^-(u)$ (similarly, $F_{S(\bm Y)}^-(v)=F_Y^-(v)$). Therefore,
      \begin{align*}
        &\phantom{{}={}}C_{S(\bm{X}),S(\bm{Y})}(u,v)=F_{\bm{X},\bm{Y}}(F_{S(\bm{X})}^-(u),\dots,F_{S(\bm{X})}^-(u),F_{S(\bm{Y})}^-(v),\dots,F_{S(\bm{Y})}^-(v))\\
        &=F_{\bm{X},\bm{Y}}(F_{X}^-(u),\dots,F_{X}^-(u),F_{Y}^-(v),\dots,F_{Y}^-(v))\\
        &=\P(F_{X}^{-}(U)\leq F_{X}^{-}(u),\dots,F_{X}^{-}(U)\leq F_{X}^{-}(u),\ F_{Y}^{-}(U)\leq F_{Y}^{-}(v),\dots,F_{Y}^{-}(U)\leq F_{Y}^{-}(v))\\
        &=\P(U \leq u,\ U\leq v)=\P(U\leq\min\{u,v\})=\min\{u,v\}.
      \end{align*}
    \item Similarly as in Part~\ref{prop:max:coll:cop:2}. \qedhere
    \end{enumerate}
  \end{proof}
\end{proposition}

Deriving the collapsed copula in special cases can provide a concrete
understanding of how and to what extent the maximum collapsing function
summarizes dependence between $\bm{X}$ and $\bm{Y}$ which we pre-specify. Here
is an example.

\begin{example}[Meta nested Archimedean copula model and the maximum collapsing function]
  Let
  $\bm{Z}=(\bm{X},\bm{Y})\sim
  F_{\bm{X},\bm{Y}}(\bm{x},\bm{y})=C_0\bigl(C_1(F_X(\bm{x})),C_2(F_Y(\bm{y}))\bigr)$,
  where $X_j\sim F_X$, $j\in\{1,\dots,p\}$, and $Y_k\sim F_Y$, $k\in\{1,\dots,q\}$.
  Furthermore, interpret $F_X(\bm{x})$ and $F_Y(\bm{y})$ as
  $(F_X(x_1),\dots,$ $F_X(x_p))$ and $(F_Y(y_1),\dots,F_Y(y_q))$,
  respectively. Let $C_0,C_1,C_2$ be Archimedean copulas with
  generators $\psi_0,\psi_1,\psi_2$, respectively, satisfying the sufficient
  nesting condition; see \cite{mcneil2008} or \cite{hofert2012b} for more
  details. Consider the maximum collapsing function $S$. Since
  \begin{align*}
    S(\bm{X})&\sim F_{S(\bm{X})}(\bm{x})=C_1(F_X(\bm{x}))=\psi_1\Bigl(\,\sum_{j=1}^p\psi_1^{-1}(F_X(x_j))\Bigr),\\
    S(\bm{Y})&\sim F_{S(\bm{Y})}(\bm{y})=C_2(F_Y(\bm{y}))=\psi_2\Bigl(\,\sum_{k=1}^q\psi_2^{-1}(F_Y(y_k))\Bigr),
  \end{align*}
  with diagonals
  $F_{S(\bm{X})}(x,\dots,x)=\psi_1\bigl(p\psi_1^{-1}(F_X(x))\bigr)$ and
  $F_{S(\bm{Y})}(y,\dots,y)=\psi_2\bigl(q\psi_2^{-1}(F_Y(y))\bigr)$, the
  corresponding quantile functions are
  \begin{align*}
    F_{S(\bm{X})}^-(u)=F_X^-\bigl(\psi_1(\psi_1^{-1}(u)/p)\bigr),\quad F_{S(\bm{Y})}^-(v)=F_Y^-\bigl(\psi_2(\psi_2^{-1}(v)/q)\bigr),
  \end{align*}
  respectively. Proposition~\ref{prop:S:max:cop} implies that the collapsed
  copula equals
  \begin{align*}
    C_{S(\bm{X}),S(\bm{Y})}(u,v)&=F_{\bm{X},\bm{Y}}(F_{S(\bm{X})}^-(u),\dots,F_{S(\bm{X})}^-(u),F_{S(\bm{Y})}^-(v),\dots,F_{S(\bm{Y})}^-(v)),\\
                                &=C_0\Bigl(C_1\Bigl(F_X\bigl(F_X^-\bigl(\psi_1(\psi_1^{-1}(u)/p)\bigr)\bigr),\dots,F_X\bigl(F_X^-\bigl(\psi_1(\psi_1^{-1}(u)/p)\bigr)\bigr)\Bigr),\\
                                &\phantom{{}=C_0\Bigl(}C_2\Bigl(F_Y\bigl(F_Y^-\bigl(\psi_2(\psi_2^{-1}(v)/q)\bigr)\bigr),\dots,F_Y\bigl(F_Y^-\bigl(\psi_2(\psi_2^{-1}(v)/q)\bigr)\bigr)\Bigr)\Bigr)\\
                                &=C_0\bigl(C_1\bigl(\psi_1(\psi_1^{-1}(u)/p),\dots,\psi_1(\psi_1^{-1}(u)/p)\bigr),\\
                                &\phantom{{}=C_2(}C_2\bigl(\psi_2(\psi_2^{-1}(v)/q),\dots,\psi_2(\psi_2^{-1}(v)/q)\bigr)\bigr)\\
                                &=C_0(u,v),\quad u,v\in[0,1].
  \end{align*}
  This is an intuitive result, as any two random variables $(X_j,Y_k)$ have
  marginal copula $C_0$ under this model and so do the group maxima (as long as
  the marginal distributions are equal per group). This implies that any
  collapsed measure of concordance is precisely the one corresponding to the
  copula $C_0$ in this case.
\end{example}

\subsection{PIT collapsing function}\label{sec:PIT:collapsed}
For the PIT collapsing function, the collapsed distribution function and copula
have notable terminology and notation following from the copula literature. In
that spirit, we will present them as extensions of the Kendall distribution
presented previously and as such adopt the same notation.
\subsubsection{Definition}
A natural extension of the univariate Kendall distribution
$K_{\bm{X}}(t)=\P(F_{\bm{X}}(\bm{X})\le t)$, $t\in[0,1]$, to the multivariate
case is the \emph{multivariate} (or \emph{joint}) \emph{Kendall distribution}, given by
\begin{align*}
  K_{\bm{X},\bm{Y}}(t_1,t_2)=\P(F_{\bm{X}}(\bm{X})\le t_1,\ F_{\bm{Y}}(\bm{Y})\le t_2)=\P(C_{\bm{X}}(\bm{U})\le t_1,\
  C_{\bm{Y}}(\bm{V})\le t_2),
\end{align*}
for all $t_1,t_2\in[0,1]$, where $\bm{U}\sim C_{\bm{X}}$ and
$\bm{V}\sim C_{\bm{Y}}$ for the copulas $C_{\bm{X}}$ and $C_{\bm{Y}}$ of
$\bm{X}$ and $\bm{Y}$, respectively; it is straightforward to define
higher-dimensional Kendall distributions. By definition,
multivariate Kendall distributions have univariate Kendall distributions as
margins. The copula of
$K_{\bm{X},\bm{Y}}(t_1,t_2)$, if uniquely determined, follows from Sklar's
Theorem via
\begin{align}
  C_K(u_1,u_2)=K_{\bm{X},\bm{Y}}(K_{\bm{X}}^-(u_1),K_{\bm{Y}}^-(u_2)),\quad u_1,u_2\in[0,1], \label{eq:Kendall_cop}
\end{align}
where $K_{\bm{X}}^-$ and $K_{\bm{Y}}^-$ denote the quantile functions of the
marginal Kendall distributions $K_{\bm{X}}$ and $K_{\bm{Y}}$, respectively. We
refer to $C_K$ as \emph{Kendall copula}. Note that Kendall copulas have previously
appeared in \cite{brechmann2014} as hierarchical Kendall copulas without
explicitly investigating the notion of joint Kendall distributions; the
latter naturally appear in our framework for measuring dependence between
random vectors.

\subsubsection{Properties}
We now briefly discuss some basic properties of multivariate Kendall
distributions and Kendall copulas. As before, we focus on the bivariate case.

As we have seen in \eqref{eq:univ:K:AC}, there is an analytical formula for
(univariate) Kendall distributions for Archimedean copulas. As we will now see,
there is also an explicit form for multivariate Kendall distributions in this
case.
\begin{proposition}[Multivariate Kendall distributions in the Archimedean case]\label{prop:joint:K:AC}
  Let $(\bm{X},\bm{Y})$ be a $(p+q)$-dimensional random vector with Archimedean
  copula $C$ with completely monotone generator $\psi$. Then, for all
  $t_1,t_2\in[0,1]$,
  \begin{align}
   \!\!\! K_{\bm{X},\bm{Y}}(t_1,t_2)=\!\!\!\!\sum_{m=0}^{(p-1)(q-1)}\!\!\!\Bigl(\,\sum_{n=0}^m\frac{(\psii(t_1))^n}{n!}\frac{(\psii(t_2))^{m-n}}{(m-n)!}\Bigr)(-1)^m\psi^{(m)}(\psii(t_1)+\psii(t_2)).\label{eq:joint:K:AC}
  \end{align}
\end{proposition}
\begin{proof}
  Let $V\sim F_V$, where $F_V$ is the Laplace--Stieltjes
  inverse of $\psi$ and let
  $E_{11},\dots,E_{1p},E_{21},$ $\dots,E_{2q}\isim\Exp(1)$.
  Furthermore, let
  \begin{align*}
    \bm{U}=\Bigl(\psi\Bigl(\frac{E_{11}}{V}\Bigr),\dots,\psi\Bigl(\frac{E_{1p}}{V}\Bigr)\Bigr)\quad\text{and}\quad
    \bm{V}=\Bigl(\psi\Bigl(\frac{E_{21}}{V}\Bigr),\dots,\psi\Bigl(\frac{E_{2q}}{V}\Bigr)\Bigr).
  \end{align*}
  Note that $(\bm{U},\bm{V})\sim C$ and that $\bm{U}\sim C_{\bm{X}}$ and
  $\bm{V}\sim C_{\bm{Y}}$, where $C_{\bm{X}},C_{\bm{Y}}$ are (also) Archimedean
  copulas with generator $\psi$. Under our assumptions, $(\bm{X},\bm{Y})$
  allows for the stochastic representation
  \begin{align*}
    (\bm{X},\bm{Y})=(F_{X_1}^-(U_{11}),\dots,F_{X_p}^-(U_{1p}),\ F_{Y_1}^-(U_{21}),\dots,F_{Y_q}^-(U_{2q})),
  \end{align*}
  and thus
  {\allowdisplaybreaks%
  \begin{align*}
    &\phantom{{}={}}K_{\bm{X},\bm{Y}}(t_1,t_2)=\P(F_{\bm{X}}(\bm{X})\le t_1,\
      F_{\bm{Y}}(\bm{Y})\le t_2)=\P(C_{\bm{X}}(\bm{U})\le t_1,\ C_{\bm{Y}}(\bm{V})\le t_2)\\
    &=\P(E_{11}+\dots+E_{1p}>V\psii(t_1),\ E_{21}+\dots+E_{2q}>V\psii(t_2))\\
    &=\int_0^\infty\P(E_{11}+\dots+E_{1p}>v\psii(t_1),\ E_{21}+\dots+E_{2q}>v\psii(t_2))\,\rd F_V(v)\\
    &=\int_0^\infty\P(E_{11}+\dots+E_{1p}>v\psii(t_1))\,\P(E_{21}+\dots+E_{2q}>v\psii(t_2))\,\rd F_V(v)\\
    &=\int_0^\infty
      F_{\Poi(v\psii(t_1))}(p-1)F_{\Poi(v\psii(t_2))}(q-1)\,\rd F_V(v)\\
    &=\int_0^\infty\exp(-v\psii(t_1))\sum_{k=0}^{p-1}\frac{(v\psii(t_1))^k}{k!}\exp(-v\psii(t_2))\sum_{l=0}^{q-1}\frac{(v\psii(t_2))^l}{l!}\,\rd F_V(v)\\
    &=\int_0^\infty\exp\bigl(-v(\psii(t_1)+\psii(t_2))\bigr)\!\!\!\sum_{m=0}^{(p-1)(q-1)}\!\!\!\Bigl(\,\sum_{n=0}^m\frac{(\psii(t_1))^n}{n!}\frac{(\psii(t_2))^{m-n}}{(m-n)!}\Bigr)v^m\,\rd F_V(v)\\
    &=\sum_{m=0}^{(p-1)(q-1)}\!\!\!\Bigl(\,\sum_{n=0}^m\frac{(\psii(t_1))^n}{n!}\frac{(\psii(t_2))^{m-n}}{(m-n)!}\Bigr)\psi^{(m)}(\psii(t_1)+\psii(t_2))(-1)^m,
  \end{align*}}%
  where we used the fact that the survival function of an Erlang distribution can be
  expressed as the distribution function $F_{\Poi}$ of a Poisson distribution.
\end{proof}
Note that \eqref{eq:univ:K:AC} follows from \eqref{eq:joint:K:AC} as a special
case. Moreover, it is straightforward to extend \eqref{eq:joint:K:AC} to higher
dimensions. In this case, each random vector in the construction
  corresponds to one dimension of the multivariate Kendall distribution. As a
  special case, when each such random vector consists of only a single random
  variable, the multivariate Kendall distribution equals the copula of these
  random variables.

Figures~\ref{fig:Kendallcop} and \ref{fig:Kendallcop2} display scatter plots of
$n=1000$ independent observations of the bivariate Gumbel and Clayton Kendall
copulas (with parameter of the underlying Gumbel and Clayton generator chosen such that Kendall's tau equals 0.5),
respectively. The different plots depict how varying dimensions $p,q$ impact the
dependence structure between the two random vectors. This difference manifests
in the form of asymmetry (lower vs upper tails) and the strength of dependence
(comparing the cases $(p,q)=(2,2)$ versus $(p,q)=50$). Furthermore, note that
there is asymmetry in the pull (that is, stronger towards the lower of $p$ and $q$
dimensions) of the realizations to the diagonal (perfect dependence).

\begin{figure}[htbp]
  \begin{center}
    \includegraphics[width=\linewidth]{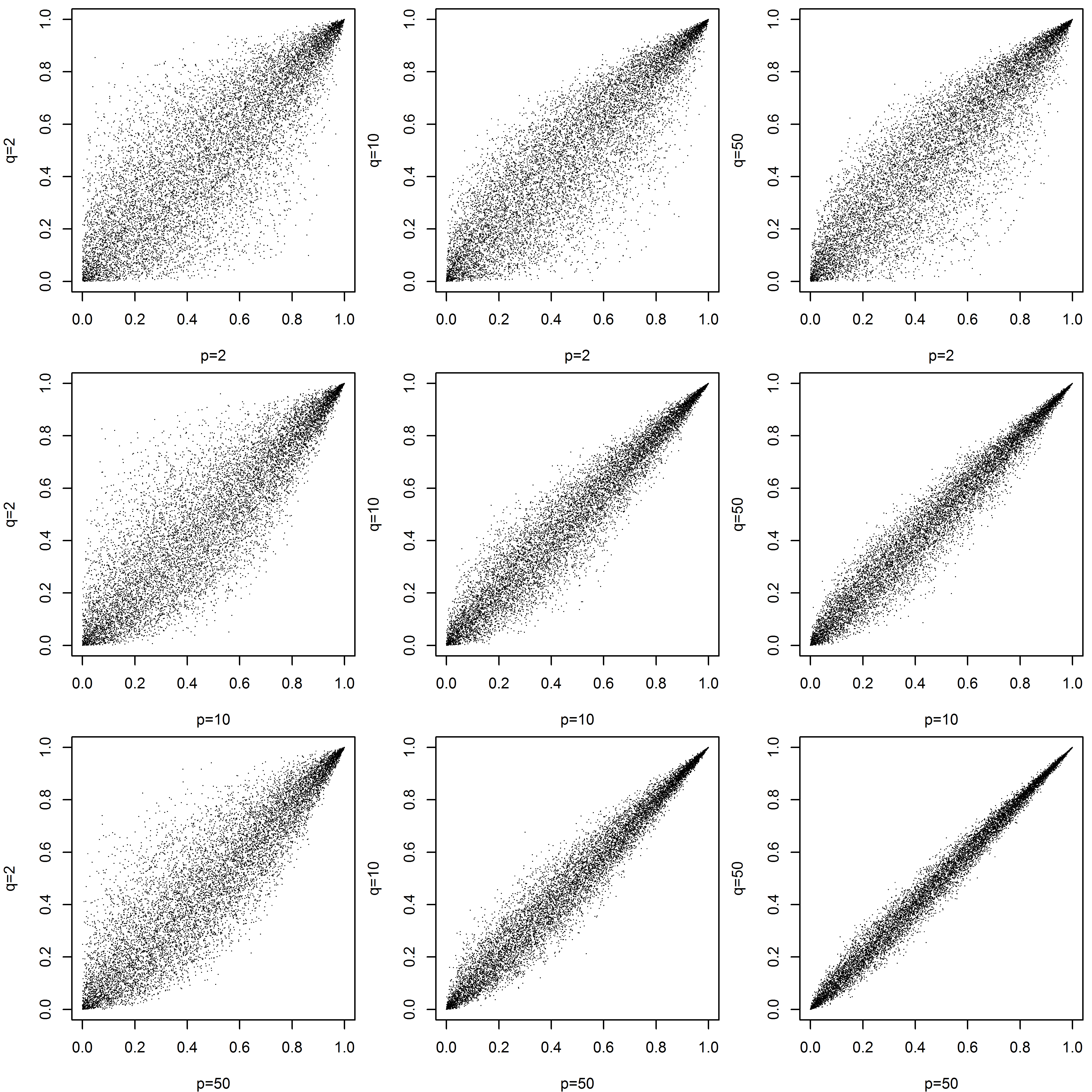}%
  \end{center}
  \caption{$n=1000$ independent observations from
    different Gumbel Kendall copulas (with Gumbel parameter chosen such that
    Kendall's tau of the underlying generator equals 0.5) corresponding to the joint Kendall
    distribution function as specified in \eqref{eq:joint:K:AC}. Note the
    dimensions of the two sectors are varied from $p \in \{2,10,50\}$ and
    $q \in \{2,10,50\}$ thus leading to nine different variations.}
  \label{fig:Kendallcop}
\end{figure}

\begin{figure}[htbp]
  \begin{center}
    \includegraphics[width=\linewidth]{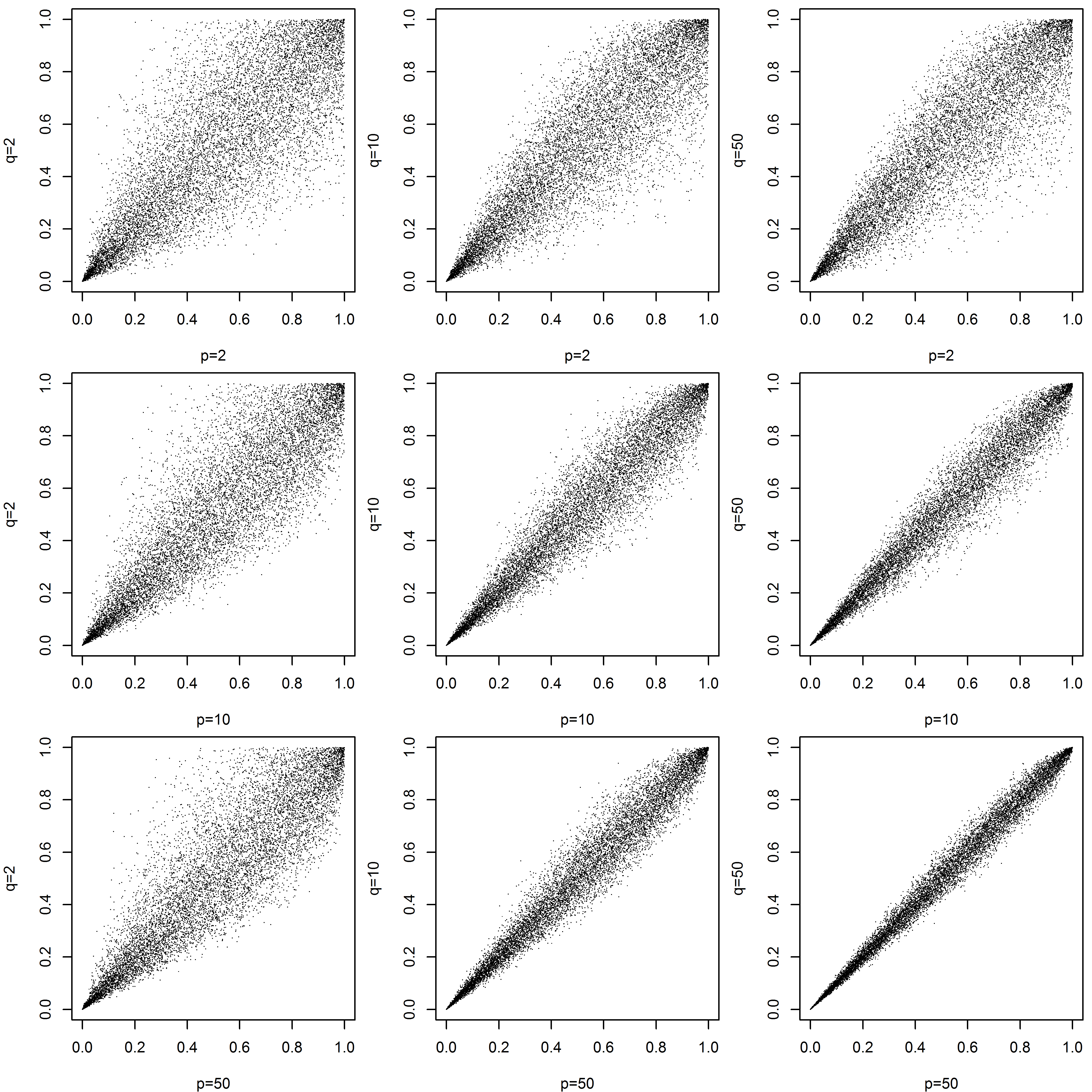}%
  \end{center}
  \caption{$n=1000$ independent observations from
    different Clayton Kendall copulas (with Clayton parameter chosen such that
    Kendall's tau of the underlying generator equals 0.5) corresponding to the joint Kendall
    distribution function as specified in \eqref{eq:joint:K:AC}. Note the
    dimensions of the two sectors are varied from $p \in \{2,10,50\}$ and
    $q \in \{2,10,50\}$ thus leading to nine different variations.}
  \label{fig:Kendallcop2}
\end{figure}

The following proposition briefly addresses basic
properties of Kendall copulas and shows that they can intuitively capture
dependence between random vectors.
\begin{proposition}[Basic properties of Kendall copulas]\label{prop:prop:Kendall:cop}
  Let $\bm{X}\sim F_{\bm{X}}$ be a $p$-dimensional and $\bm{Y}\sim F_{\bm{Y}}$
  be a $q$-dimensional random vector, both with continuously distributed
  margins (denoted as before).
  \begin{enumerate}
  \item If $\bm{X}$ and $\bm{Y}$ are independent, then
    $K_{\bm{X},\bm{Y}}(t_1,t_2)=K_{\bm{X}}(t_1)K_{\bm{Y}}(t_2)$ and thus the
    Kendall copula is the independence copula.
  \item If $\bm{X}$ and $\bm{Y}$ are comonotone,
    then $K_{\bm{X},\bm{Y}}(t_1,t_2)=\min\{t_1,t_2\}$
    and thus the Kendall copula (which equals the multivariate Kendall distribution in
    this case) is the upper Fr\'echet--Hoeffding bound.
  \item If $\bm{X}$ and $\bm{Y}$ are countermonotone,
    then $K_{\bm{X},\bm{Y}}(t_1,t_2)=\max\{t_1+t_2-1, 0\}$ and thus the Kendall
    copula (which equals the multivariate Kendall distribution in this case) is the
    lower Fr\'echet--Hoeffding bound.
  \end{enumerate}
\end{proposition}
\begin{proof}
  Let $C_{\bm{X}},C_{\bm{Y}}$ denote the copulas of $\bm{X},\bm{Y}$, respectively.
  \begin{enumerate}
  \item $K_{\bm{X},\bm{Y}}(t_1,t_2)=\P(F_{\bm{X}}(\bm{X})\le
    t_1,F_{\bm{Y}}(\bm{Y})\le t_2)\omu{\text{ind.}}{=}{}\P(F_{\bm{X}}(\bm{X})\le
    t_1)\P(F_{\bm{Y}}(\bm{Y})\le t_2)=K_{\bm{X}}(t_1)K_{\bm{Y}}(t_2)$, $t_1,t_2\in[0,1]$.
  \item\label{prop:prop:Kendall:cop:2} With $\bm{X}=(F_{X_1}^-(U),\dots,F_{X_p}^-(U))$ and
    $\bm{Y}=(F_{Y_1}^-(U),\dots,F_{Y_q}^-(U))$ for $U\sim\U(0,1)$, Sklar's
    Theorem implies that $F_{\bm{X}}(\bm{X})=C_{\bm{X}}(U,\dots,U)=\min\{U,\dots,U\}=U$ and
    $F_{\bm{Y}}(\bm{Y})=C_{\bm{Y}}(U,\dots,U)=\min\{U,\dots,U\}=U$. Therefore, $K_{\bm{X},\bm{Y}}(t_1,t_2)=\P(C_{\bm{X}}(U,\dots,$ $U)\le
    t_1,C_{\bm{Y}}(U,\dots,U)\le t_2)=\P(U\le t_1,U\le t_2)=\min\{t_1,t_2\}$. Note that
    $K_{\bm{X},\bm{Y}}(t_1,t_2)$ has $\U(0,1)$ margins in this case and thus
    equals its Kendall copula.
  \item Similarly as in Part~\ref{prop:prop:Kendall:cop:2}. \qedhere
  \end{enumerate}
\end{proof}

Nonparametric estimators of univariate Kendall distributions based on a
random sample $(\bm{X}_i,\bm{Y}_i)$, $i\in\{1,\dots,n\}$, can be constructed as follows.
Let
\begin{align*}
  \bm{W}_i=(W_{i1},W_{i2})=\Bigl(\frac{1}{n-1}\sum_{\substack{k=1\\ k\neq i}}^{n}\I_{\{\bm{X}_k\le\bm{X}_i\}},\ \frac{1}{n-1}\sum_{\substack{k=1\\ k\neq i}}^{n}\I_{\{\bm{Y}_k\le\bm{Y}_i\}}\Bigr),
\end{align*}
where, as usual, the inequalities are understood componentwise. Analogously to
\cite{genestrivest1993} and \cite{barbegenestghoudiremillard1996} in the
univariate case, one can use the multivariate empirical distribution function
\begin{align*}
  K_n(\bm{t})=K_n(t_1,t_2)=\frac{1}{n}\sum_{i=1}^n\I_{\{\bm{W}_i\le \bm{t}\}}=\frac{1}{n}\sum_{i=1}^n\I_{\{W_{i1}\le t_1, W_{i2}\le t_2\}},\quad \bm{t}=(t_1,t_2)\in[0,1]^2,
\end{align*}
as a nonparametric estimator of $K_{\bm{X},\bm{Y}}(t_1,t_2)$.

\subsubsection{Dependence measures related to the multivariate Kendall
  distribution}

We now turn to a link between multivariate Kendall distributions and
dependence measures of the form $\chi(\bm{X},\bm{Y})$.  Below are a few examples
starting with the dependence measure resulting from the PIT collapsing function.
\begin{example}[Correlation via the joint Kendall distribution]\label{example:correlation_jointKendall}
  Since $K_{\bm{X}}(t_1)$ and $K_{\bm{Y}}(t_2)$ are the distribution functions of
  $F_{\bm{X}}(\bm{X})$ and $F_{\bm{Y}}(\bm{Y})$, respectively, and $K_{\bm{X},\bm{Y}}(t_1,t_2)$ is
  the joint distribution function of
  $(F_{\bm{X}}(\bm{X}),F_{\bm{Y}}(\bm{Y}))$, Hoeffding's Identity implies that
  \begin{align*}
    \chi(\bm{X},\bm{Y})&=\rho(F_{\bm{X}}(\bm{X}),F_{\bm{Y}}(\bm{Y}))=\frac{\Cov (F_{\bm{X}}(\bm{X}),F_{\bm{Y}}(\bm{Y}))}{\sqrt{\Var(F_{\bm{X}}(\bm{X}))\Var(F_{\bm{Y}}(\bm{Y}))}}\\
    &=\frac{\iint_{[0,1]^2} K_{\bm{X},\bm{Y}}(t_1,t_2)-K_{\bm{X}}(t_1)K_{\bm{Y}}(t_2) \,\rd t_1\rd t_2}{\sqrt{\int_{[0,1]}K_{\bm{X}}(t_1)-K^2_{\bm{X}}(t_1)\,\rd t_1\int_{[0,1]}K_{\bm{Y}}(t_2)-K^2_{\bm{Y}}(t_2)\,\rd t_2}}.
  \end{align*}
  Note that the numerator is the (integrated) difference between the joint Kendall
  distribution of $\bm{X}$ and $\bm{Y}$ and the joint Kendall
  distribution under independence of $\bm{X}$ and $\bm{Y}$;
  $\chi(\bm{X},\bm{Y})$ thus represents in some sense how far on average the random vectors $\bm{X}$ and
  $\bm{Y}$ are from independence, thus mimicking the construction of
  standard bivariate dependence measures.
\end{example}

\begin{example}[Spearman's rho via the joint Kendall distribution]\label{example: spearman_jointKendall}
  One drawback of the measure presented in Example
  \ref{example:correlation_jointKendall} is that it depends on the marginal
  distributions of the collapsed random variables. To rectify this, we can apply
  the marginal Kendall distributions $K_{\bm{X}}$ and $K_{\bm{Y}}$ to the
  collapsed random variables $F_{\bm{X}}(\bm{X})$ and $F_{\bm{Y}}(\bm{Y})$,
  respectively. The measure will then be a natural multivariate extension of
  Spearman's rho as it only depends on the Kendall copula. To this end, let
  $U=K_{\bm{X}}(F_{\bm{X}}(\bm{X}))$ and
  $V=K_{\bm{Y}}(F_{\bm{Y}}(\bm{Y}))$. Then,
  \begin{align*}
    \chi(\bm{X},\bm{Y})&=\rho\bigl(K_{\bm{X}}(F_{\bm{X}}(\bm{X})),\ K_{\bm{Y}}(F_{\bm{Y}}(\bm{Y}))\bigr)=\rho(U,V)=\frac{\E[UV] -1/4}{1/12}=12\E[UV]-3\\
    &=12 \iint_{[0,1]^2} uv\,\rd C_K(u,v)-3=\rho_S(F_{\bm{X}}(\bm{X}),F_{\bm{Y}}(\bm{Y})),
  \end{align*}
  where $C_K(u,v)$ denotes the Kendall copula introduced
  in \eqref{eq:Kendall_cop}. Thus, $\chi(\bm{X},\bm{Y})$ equals Spearman's rho
  of $F_{\bm{X}}(\bm{X})$ and $F_{\bm{Y}}(\bm{Y})$.
\end{example}

\begin{example}[Kendall's tau via the joint Kendall distribution]
  Similarly, with $U$ and $V$ as defined in Example~\ref{example: spearman_jointKendall}, for Kendall's tau we have
  \begin{align*}
    \!\!\!\chi(\bm{X},\bm{Y})&=\rho(\I_{\{F_{\bm{X}}(\bm X)\le F_{\bm{X}}(\bm X')\}},\I_{\{F_{\bm{Y}}(\bm Y)\le F_{\bm{Y}}(\bm Y')\}}) =\tau(F_{\bm{X}}(\bm{X}),F_{\bm{Y}}(\bm{Y}))\\
    &=4 \iint_{[0,1]^2} C_K(u,v)\,\rd C_K(u,v)-1,
  \end{align*}
  where $C_K(u,v)$ denotes the Kendall copula as before. The first equality
  follows from Lemma \ref{lem:repr:tau:cor} and the last equality follows by
  definition of Kendall's tau of the collapsed random variables in the bivariate
  case. This measure forms a multivariate extension of Kendall's tau which only
  depends on the Kendall copula. An alternative albeit very similar extension
  was formed via the multivariate rank collapsing function presented in Section~\ref{subsec:multirank}.
\end{example}

\begin{example}[Tail dependence via Kendall copulas] \label{ex:tail_Kendallcop}
  In light of using \eqref{chi:lam} for measuring tail dependence between the
  collapsed random variables, it is easy to see that when using the PIT
  collapsing function, \eqref{chi:lam} as measure of association corresponds to
  computing (classical) coefficients of tail dependence of the underlying
  Kendall copula $C_K$. For example, if $\bm{X}\sim F_{\bm{X}}$, $\bm{Y}\sim
  F_{\bm{Y}}$ with Kendall distributions $K_{\bm{X}}$, $K_{\bm{Y}}$,
  respectively, and if $U=K_{\bm{X}}(F_{\bm{X}}(\bm{X}))$,
  $V=K_{\bm{Y}}(F_{\bm{Y}}(\bm{Y}))$ (note that $(U,V)\sim C_K$ in this case),
  then the coefficient of upper tail dependence can be expressed as
  \begin{align*}
    \lambda_{U}(F_{\bm{X}}(\bm{X}),F_{\bm{Y}}(\bm{Y}))&=\lim_{u
      \uparrow 1} \P(F_{\bm{Y}}(\bm{Y}) > K_{\bm{Y}}^{-}(u)\,|\,F_{\bm{X}}(\bm{X}) > K_{\bm{X}}^{-}(u))\\
    &=\lim_{u \uparrow 1}\P(V>u\,|\,U>u)= \lim_{u \uparrow 1}\frac{1-2u + C_K(u,u)}{1-u}.
  \end{align*}
\end{example}

\section{Estimation and properties}\label{sec:fit:S}
In this section, we study sample estimators of Equation~\eqref{eq:measure}, and
derive asymptotic results which can be used to compute their standard errors.

\subsection{Estimation for general collapsing functions $S$}
Assume we have a random sample $(\bm{X}_1,\bm{Y}_1),\dots,(\bm{X}_n,\bm{Y}_n)$
from $F_{\bm{X},\bm{Y}}$. An estimator $\chi_n$ of
$\chi(\bm{X},\bm{Y})=\rho(S(\bm X), S(\bm Y))$ can be constructed by replacing
$\rho$ by the sample correlation coefficient. The following section investigates some properties of this estimator for general $S$ (but not the PIT
collapsing function which, due to its nature, is treated in the section thereafter).

\subsection{Asymptotic result for general $S$}%\label{asymptotics}
In this section we closely follow \cite{grothe2014} in order to construct $\chi_n$ through the lens of U-statistics for deriving its asymptotic distribution. To this end we have the
following proposition.
\begin{proposition}[Asymptotic distribution of $\chi_n$]\label{Prop:Asymptotics}
  Suppose $\chi_n(\bm{X},\bm{Y})$ is defined as above. Then, as $n\to\infty$,
  \begin{align*}
    \sqrt{n} ({\chi}_n-\chi)\overset {d} {\longrightarrow} \N(0,\sigma^2_{\chi}),
  \end{align*}
  where
  \begin{align*}
    \sigma_{\chi}^2=\begin{cases}
      (\nabla f_{5 \times 1}|_{\bm{\mu}})'\Sigma_1(\nabla f_{5 \times 1}|_{\bm{\mu}}),&\text{if $S$ is a $p$-variate function},\\
      4(\nabla f_{5 \times 1}|_{\bm{\mu}})'\Sigma_2(\nabla f_{5 \times 1}|_{\bm{\mu}}),&\text{if $S$ is a $2p$-variate function}.
    \end{cases}
  \end{align*}
  Here, $\nabla f_{5 \times 1}|_{\bm{\mu}}$ denotes the gradient vector
  of the function
  \begin{align*}
    f(a,b,c,d,e)=\frac {e-ab} {\sqrt{c-a^2} \sqrt{d-b^2}},
  \end{align*}
  evaluated at the population mean
  $\bm{\mu}=(\mu_x,\mu_y,\mu_{xx},\mu_{yy},\mu_{xy})$, where
  $\mu_x=\E[S(\bm{X})]$, $\mu_y=\E[S(\bm{Y})]$, $\mu_{xx}=\E[S(\bm{X})^2]$,
  $\mu_{yy}=\E[S(\bm{Y})^2]$, $\mu_{xy}=\E[S(\bm{X})S(\bm{Y})]$.
  Furthermore, $\Sigma_1$ denotes the covariance matrix of
  $(S(\bm{X}),S(\bm{Y}),S(\bm{X})^{2},S(\bm{Y})^{2},S(\bm{X})S(\bm{Y}))$ and
  $\Sigma_2$ denotes the covariance matrix of
  $(\E_{\tiny \bm{X}'}[S(\bm{X},\bm{X}')],$ $\E_{\tiny
    \bm{Y}'}[S(\bm{Y},\bm{Y}')],$ $\E_{\tiny
    \bm{X}'}[S(\bm{X},\bm{X}')^2],$ $\E_{\tiny
    \bm{Y}'}[S(\bm{Y},\bm{Y}')^2],$ $\E_{\tiny
    (\bm{X}',\bm{Y}')}[S(\bm{X},\bm{X}')S(\bm{Y},\bm{Y}')])$.
\end{proposition}
\begin{proof}
  Refer to Appendix \ref{proof:Prop 4.1} for the details.
\end{proof}

\begin{remark}[Estimation of $\sigma^2_{\chi}$]\label{remark:est}
  To estimate the asymptotic variance $\sigma_{\chi}^2$ we adopt a plug-in
  approach as was suggested by \cite{grothe2014}. This procedure has
    two key ingredients as summarized below and it will slightly differ between
    the two cases given in the proof of
    Proposition~\ref{Prop:Asymptotics}. Further, note that the notation below is
    also explained in the proof of Proposition~\ref{Prop:Asymptotics} in
    Appendix~\ref{proof:Prop 4.1}. The two cases referred to below correspond
  to when $S$ is a $p$-variate (Case~1) or a $2p$-variate (Case~2) function.
  \begin{enumerate}
  \item In a first step, evaluate the gradient vector $\nabla f_{5 \times 1}$ at
    $\bm{m^{(k)}}=(m^{(k)}_x,m^{(k)}_y,m^{(k)}_{xx},m^{(k)}_{yy},m^{(k)}_{xy})$,
    $k\in\{1,2\}$ corresponding to the sample quantities in Case~$k$. The
    analytical form of the gradient vector evaluated at the
    appropriate values is given in Appendix~\ref{App:add_details}.
  \item Now distinguish the two cases: In Case~1, estimate $\Sigma_1$ by the sample covariance matrix $\Sigma_{n,1}$ of $\big(S(\bm{X}_i),S(\bm{Y}_i),
    S(\bm{X}_i)^{2},S(\bm{Y}_i)^2,S(\bm{X}_i)S(\bm{Y}_i)\big)$,
    $i\in\{1,\dots,n\}$. In Case~2, estimate $\Sigma_2$ by the sample
    covariance matrix $\Sigma_{n,2}$ of
    $\big(g_x(\bm{X}_i),g_y(\bm{Y}_i),g_{xx}(\bm{X}_i),g_{yy}(\bm{Y}_i),g_{xy}(\bm{X}_i,\bm{Y}_i)\big)$,
    $i\in\{1,\dots,n\}$, where
    \begin{align*}
      &g_{x}(\bm{X}_i)=\frac {1}{n-1} \sum_{\substack{j=1\\ j\ne i}}^n S(\bm{X}_i,\bm{X}_j),\quad g_{y}(\bm{Y}_i)=\frac {1}{n-1} \sum_{\substack{j=1\\ j\ne i}}^n S(\bm{Y}_i,\bm{Y}_j),\\
      &g_{xx}(\bm{X}_i)=\frac {1}{n-1} \sum_{\substack{j=1\\ j\ne i}}^n S(\bm{X}_i,\bm{X}_j)^2,\quad g_{yy}(\bm{Y}_i)=\frac {1}{n-1} \sum_{\substack{j=1\\ j\ne i}}^n S(\bm{Y}_i,\bm{Y}_j)^2, \\
      &g_{xy}(\bm{X}_i,\bm{Y}_i)=\frac {1}{n-1} \sum_{\substack{j=1\\ j\ne i}}^n S(\bm{X}_i,\bm{X}_j)S(\bm{Y}_i,\bm{Y}_j).
    \end{align*}
    The quantities $g_x$, $g_y$, $g_{xx}$, $g_{yy}$, $g_{xy}$ estimate the
    conditional expectations, $(\E_{\tiny \bm{X}'}[S(\bm{X},\bm{X}')],$ $\E_{\tiny
    	\bm{Y}'}[S(\bm{Y},\bm{Y}')],$ $\E_{\tiny
    	\bm{X}'}[S(\bm{X},\bm{X}')^2],\E_{\tiny
    	\bm{Y}'}[S(\bm{Y},\bm{Y}')^2],\E_{\tiny
    	(\bm{X}',\bm{Y}')}[S(\bm{X},\bm{X}')S(\bm{Y},\bm{Y}')])$, and can be motivated from a
    jackknife methodology as \cite{grothe2014} identified.
  \item Then,
    $\sigma^2_{n,\chi}=(\nabla f\rvert_{\bm{m}^{(1)}})'\Sigma_{n,1}(\nabla f\rvert_{\bm{m}^{(1)}})$ in Case~1 and
    $\sigma^2_{n,\chi}=4(\nabla
    f\rvert_{\bm{m}^{(2)}})'\Sigma_{n,2}(\nabla
    f\rvert_{\bm{m}^{(2)}})$ in Case~2.
  \end{enumerate}
\end{remark}

\subsection{Estimation for the PIT collapsing function}
We now discuss the construction of an estimator for
$\chi(\bm{X},\bm{Y})=\rho(F_{\bm{X}}(\bm{X}),F_{\bm{Y}}(\bm{Y}))$. In contrast
to the asymptotic U-statistics framework developed earlier, we cannot
directly express $\chi(\bm{X},\bm{Y})$ as a function of
U-statistics in this case.

To begin with, let $W_1=F_{\bm{X}}(\bm{X})$ and $W_2=F_{\bm{Y}}(\bm{Y})$. As
in \cite{barbegenestghoudiremillard1996}, we consider the
pseudo-observations
\begin{align*}
  W_{i1}=\frac{1}{n-1}\sum_{\substack{k=1\\ k\neq i}}^{n}\I_{\{\bm{X}_k\le\bm{X}_i\}},\quad W_{i2}=\frac{1}{n-1}\sum_{\substack{k=1\\ k\neq i}}^{n}\I_{\{\bm{Y}_k\le\bm{Y}_i\}}, \quad i\in\{1,\dots,n\},
\end{align*}
where the inequalities are understood componentwise. Similarly as before, an
estimator for the dependence measure $\chi(\bm{X},\bm{Y})$ can simply be
constructed via the sample correlation coefficient, that is,
$\chi_n(\bm{X},\bm{Y})=\rho_n(W_{i1},W_{i2})$. As this particular estimator does
not fit in the U-statistics framework, it is harder to derive
asymptotic normality with an expression for the asymptotic variance
for this collapsing function. However, we can always construct bootstrap
confidence intervals if required.

Based on the pseudo-observations defined above, one can also estimate
$\chi(\bm{X},\bm{Y})=\rho_S(F_{\bm{X}}(\bm{X}),F_{\bm{Y}}(\bm{Y}))$ by
$\chi_n(\bm{X},\bm{Y})=\rho_n(K_{n,\bm{X}}(W_{i1}),K_{n,\bm{Y}}(W_{i2}))$, where
\begin{align*}
  K_{n,\bm{X}}(t_1)=\frac{1}{n}\sum_{j=1}^{n} \I_{\{W_{i1} \le t_1\}}, \quad K_{n,\bm{Y}}(t_2)=\frac{1}{n}\sum_{j=1}^{n} \I_{\{W_{i2} \le t_2\}}, \quad t_1,t_2 \in [0,1]. %\label{eq:marginalkednall}
\end{align*}

\section{Applications}\label{sec:ex}
We now present two applications, one in bioinformatics
(see Section~\ref{sec:ex:prot}) and another in finance (see Section~\ref{sec:ex:SP500}).

\subsection{Protein data: An application from bioinformatics}\label{sec:ex:prot}
\subsubsection{Introduction}

Proteins are complex molecules composed of sequences of amino acid
residues. There are 20 different types of amino acids. All of them have the same
generic structure, R-CH(NH$_2$)-COOH, where the component labelled ``R'', also
known as a side chain, identifies the specific type of amino acid. In
bioinformatics, scientists are interested in understanding how conformational
changes at different side chains may be coupled together; see, for example,
\cite{ghoraie2015}. For example, if two residues are far apart in the sequence
but their side chains tend to change conformation together, it may be an
indication that they are close in 3D. In turn, this may shed light on the
all-important underlying protein folding process.

The conformation of a side chain can be characterized by a set of dihedral
angles. To understand this, picture a side chain as a sequence of atoms spanning
off the backbone of the protein. The angle between planes formed by atoms 1--3
and atoms 2--4 in the sequence is referred to as the \emph{first}
dihedral angle, and so on. Typically, there are zero to four such dihedral angles
depending on the size of the underlying amino acid.

Thus, let $\bm{X}= (X_1,\dots,X_p)$, $0 \le p \le 4$, and
$\bm{Y}=(Y_1,\dots,Y_q)$, $0\le q \le 4$, represent the dihedral angles of two
side chains, respectively. We are then in need of a measure of
dependence between the two random vectors $\bm{X}$ and $\bm{Y}$. To quantify their
dependence, \cite{ghoraie2015sparse} applied the Graphical LASSO (GLASSO)
developed by \cite{friedman2008}, while \cite{ghoraie2015} used ``kernelized
partial canonical correlation analysis'' (KPCCA). Here, we apply
our framework of collapsing functions.

\subsubsection{Analysis}
Below, we will report results using various collapsing functions -- in
particular, the weighted average, the pairwise distance, the pairwise kernel,
and the PIT.

For the weighted average, we put more weight on the first few dihedral angles,
starting with the extreme case of $\bm{w}=(1,0,...,0)$, that is, full weight on
the first dimension. This is because, biologically, the dihedral angles
closer to the backbone of the protein are more meaningful than those
further away.

For the pairwise distance, we include only the Euclidean distance because, after experimenting with other distance functions, there was little to no difference for this application.

For the pairwise kernel, we follow \cite{ghoraie2015} and use a multivariate von-Mises kernel,
\begin{align*}
K(\bm{x}_i,\bm{x}_j)= \prod_{t=1}^{p} \exp(\kappa_t\cos(x_{it}-x_{jt})),
\end{align*}
where $\bm{x}_i, \bm{x}_j \in \IR^p$ are two different conformations of a given side chain.
We simply use the same concentration parameters as those adopted and justified by \cite{ghoraie2015}, so $\kappa_1=8$, $\kappa_2=8$, $\kappa_3=4$ and $\kappa_4=2$. These choices were because atoms farther away from the backbone have more freedom of motion.

Finally, the PIT is a general purpose choice of collapsing function
that can capture both positive and negative association. However, for the
purpose of ranking dependencies we are just interested in the strength (not the
direction) of dependence, so we use $|\chi(\bm{X},\bm{Y})|$ as the ranking criteria.

We use the same dataset as in \cite{ghoraie2015} which allows for
a direct comparison of the results. Altogether,
\cite{ghoraie2015} studied eight different types of proteins from three
different families (Ras, Rho and Rab). Each protein has a varying number of
residues approximately in the range of 160--190. Through a specific procedure
explained in \cite{ghoraie2015}, roughly 16,000--18,000 sample conformations for
these proteins were generated.

 Note that working with the pair-wise distance and kernel
collapsing functions is computationally prohibitive. Since, for each protein, we
have up to 18,000 sample conformations, this would have resulted in
$18,000 \choose 2$ samples in the collapsed space. We thus consider ten random
subsets of size 2500 (without replacement to avoid pair-wise distance or kernel
similarity values of zero) from the original dataset and compute the relevant
evaluation criteria as an average across the obtained subsets.

The objective here is to rank all pairs of residues in a protein according to
various dependence measures, and to verify whether ``known couplings'' appear in
the top-ranked pairs. Following \cite{ghoraie2015}, ``known couplings'' were
based on the Contact Rearrangement Network (CRN) method from \cite{daily2008}. The receiver-operating characteristic (ROC) curve -- in
particular, the area under the ROC curve (AUC) -- is used as a summarizing
evaluation criterion to determine how well the rankings produced by different
dependence measures agree with the CRN method's results.

\subsubsection{Results and discussion}

We compare the resulting AUC values from the chosen collapsing functions with results from KPCCA (see \cite{ghoraie2015}) and GLASSO
(see \cite{ghoraie2015sparse}). This comparison is summarized in Table~\ref{tab:AUC_protein}.
\begin{table}[htbp]
  \centering
  \begin{tabular}{@{\extracolsep{-1pt}}lcccccccc}
    \toprule
    \multicolumn{1}{c}{Protein} & \multicolumn{1}{c}{H-Ras}   & \multicolumn{1}{c}{RhoA}   & \multicolumn{1}{c}{Rap2A}   & \multicolumn{1}{c}{Rheb}   & \multicolumn{1}{c}{Sec4}  & \multicolumn{1}{c}{Cdc42}   & \multicolumn{1}{c}{Rac1}   & \multicolumn{1}{c}{Ypt7p} \\
    \multicolumn{1}{c}{PBD ID} & \multicolumn{1}{c}{4Q21} &  \multicolumn{1}{c}{1FTN} &  \multicolumn{1}{c}{1KAO} &   \multicolumn{1}{c}{1XTQ} &   \multicolumn{1}{c}{1G16} &   \multicolumn{1}{c}{1ANO} &  \multicolumn{1}{c}{1HH4(A)} & \multicolumn{1}{c}{1KY3} \\
    \midrule
    KPCCA      &                 80 & 75 &  69 & 70 & 68 &  68 &    67 &  72\\
    Weighted Average  &    78 & 74 &  72 & 65 & 71 &  66 &    67 &  59\\
    GLASSO      &              78 & 72 &  68 & 71 & 68 &  68 &    59 &  67\\
    Kernel &                       78 & 72 &  71 & 69 & 66 &  65 &    63 &  61\\
    Distance      &               77 & 72 &  71 & 69 & 65 &  65 &    63 &  61\\
    PIT &                            74 & 68 &  70 & 71 & 68 &  61 &    59 &  57\\
    \bottomrule
  \end{tabular}
  \caption{AUC (values in percent) against CRN. The rows and columns are
    organized in decreasing order of row and column means. Note that the PDB ID
    is a unique identifier of the inactive state of the protein; see
    \cite{berman2006}.}
  \label{tab:AUC_protein}
\end{table}

Firstly, from all AUC values in Table~\ref{tab:AUC_protein},
we see that dependence measures resulting from all collapsing functions often
possessed significantly better allosteric coupling detection power than at
random (when AUC is 50\%). Furthermore, very simple yet meaningful collapsing
functions, such as the weighted average (with $\bm{w}=(1,0,...,0)$), often
yielded comparable (for 1G16 and 1KAO even better) results to KPCCA and
GLASSO. This is an interesting observation, given that this particular
collapsing function is considerably faster and easier to understand than the
mathematically sophisticated KPCCA or GLASSO methods.

\subsection{S\&P~500: An application from finance}\label{sec:ex:SP500}
\subsubsection{Introduction}
There are numerous problems in finance and risk management that require to study
the dependence between random vectors or groups of random variables. In this
section, we explore such a problem in the setting of investigating dependence
between S\&P~500 business sectors. Furthermore, as we are dealing with time
series data, this problem can be viewed both through the prism of static and
dynamic dependence. Fixing a time period, we can assess whether the business
sectors are independent by visualizing the dependence between
them. Additionally, we can compute time-varying dependence measures to
dynamically capture dependence between business sectors.

\subsubsection{S\&P~500 constituent data}
For the static case, we consider the 465 available constituent time series from
the S\&P~500 in the time period from 2007-01-01 to 2009-12-31 (756
  trading days); see the \R\ package \texttt{qrmdata}. For the dynamic case, we
consider 461 (due to missing data) of these constituent time
series. We use the 10 Global Industry Classification Standard (GICS) sectors as
business sectors. Nine GICS sectors have Exchange Traded Funds (ETFs) which
track the performance of each business sector. These marketable securities are
referred to as sector SPDR ETFs. We use a bivariate measure of dependence
between any two sector ETFs as a market-determined benchmark for
comparisons.

Turning to pre-processing of the dataset, we work with negative log-returns for
each constituent. Furthermore, we fit ARMA(1,1)-GARCH(1,1) models to each time
series and extract the corresponding standardized residuals to investigate
dependence between the component series; see \cite{patton2006} for this
procedure. Note that we apply the same pre-processing to the nine ETF time series.
\subsubsection{A snapshot of S\&P~500 sector dependence}
Following Algorithm \ref{vis:test:indep}, we can perform an assessment of
independence between business sectors. In particular, we use (Euclidean)
distance, weighted average (equal weights), maximum, and PIT collapsing
functions. We also visualize the dependence between all 36~ETF sector pairs
(notice the Telecommunications sector does not have an ETF) for comparison.
\begin{figure}[htbp]
  \begin{center}
    \includegraphics[width=0.45\linewidth]{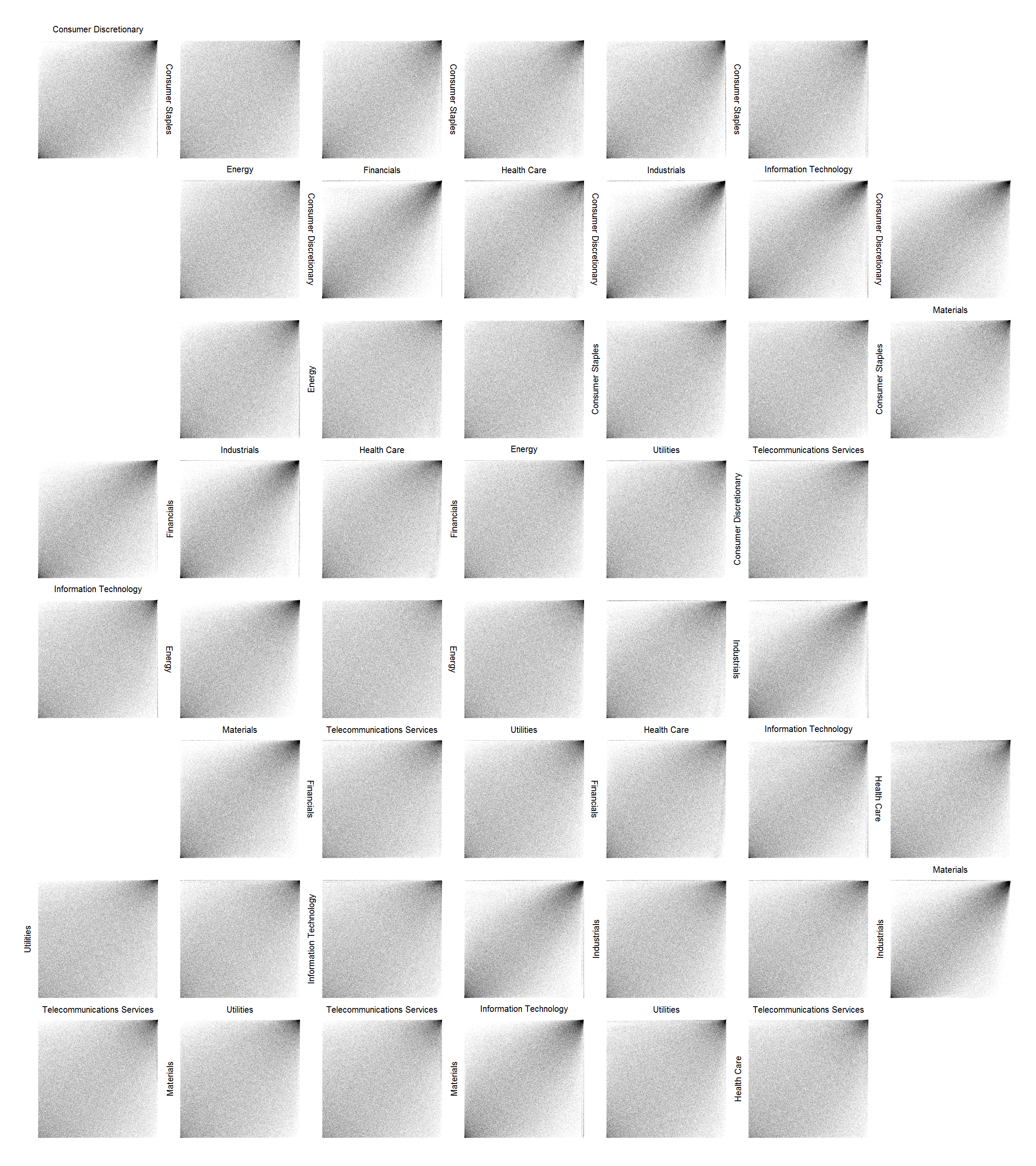}%
    \hspace{1mm}
    \includegraphics[width=0.45\linewidth]{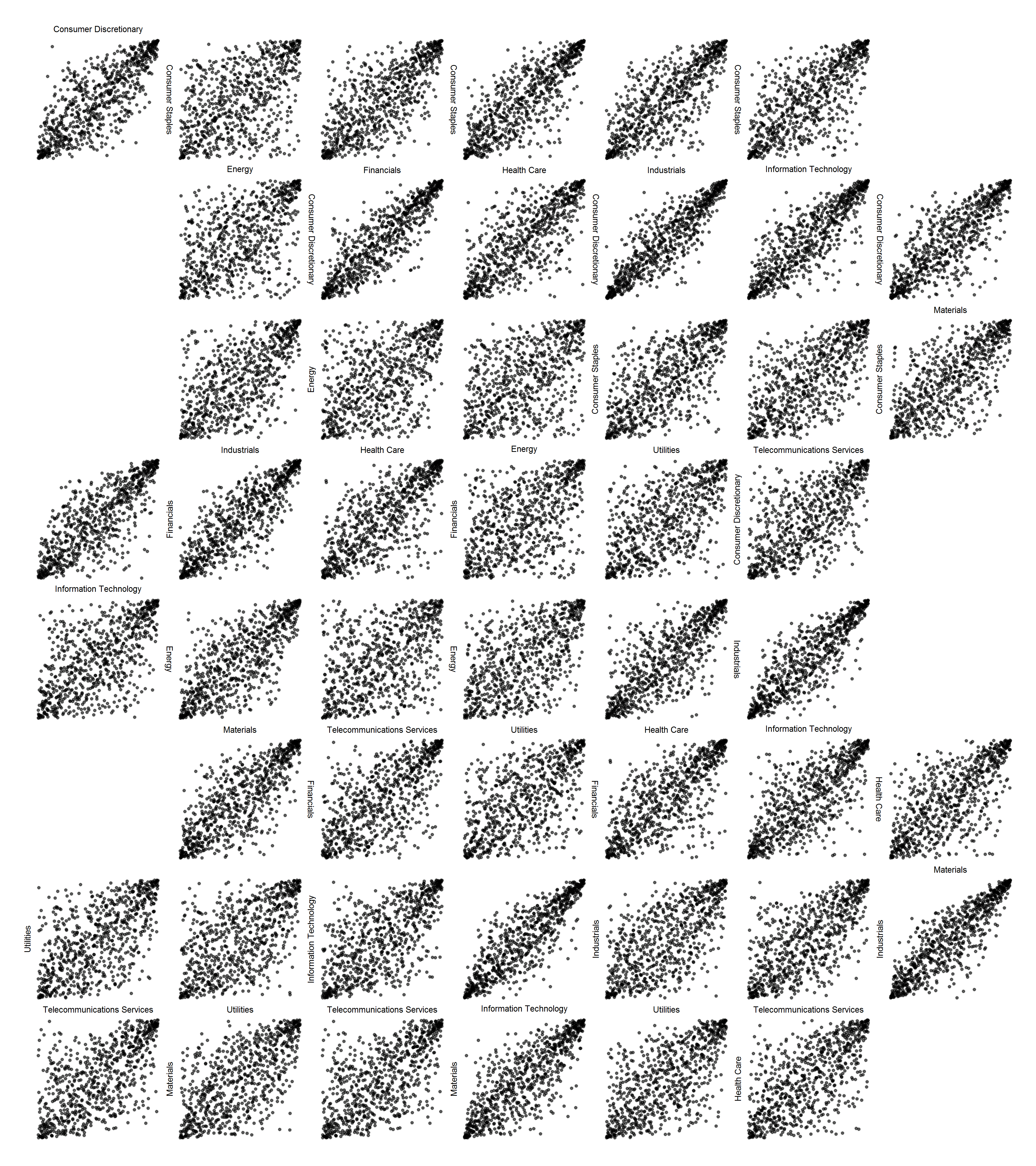}\\

    \vspace{1mm}
    \includegraphics[width=0.45\linewidth]{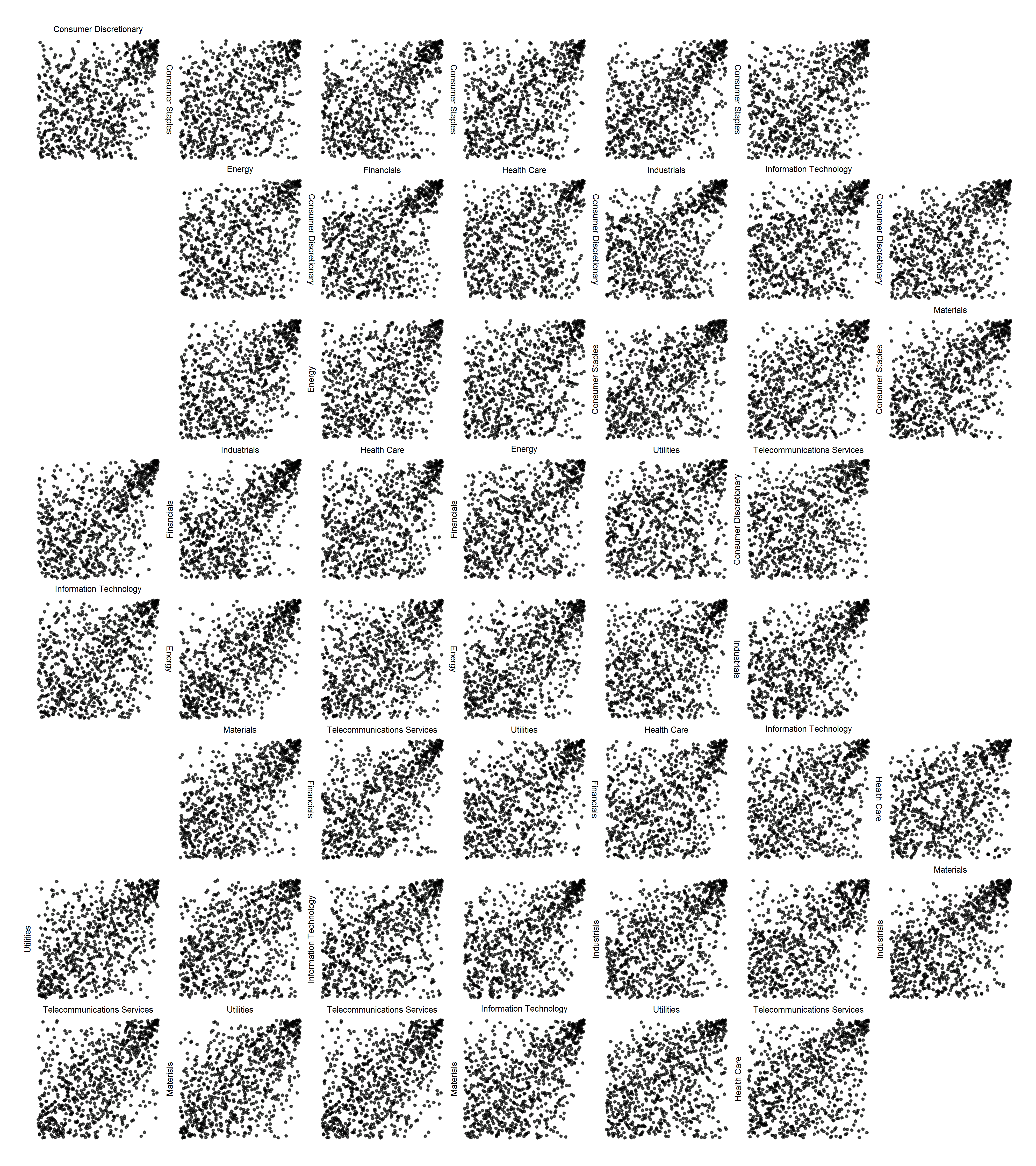}%
    \hspace{1mm}
    \includegraphics[width=0.45\linewidth]{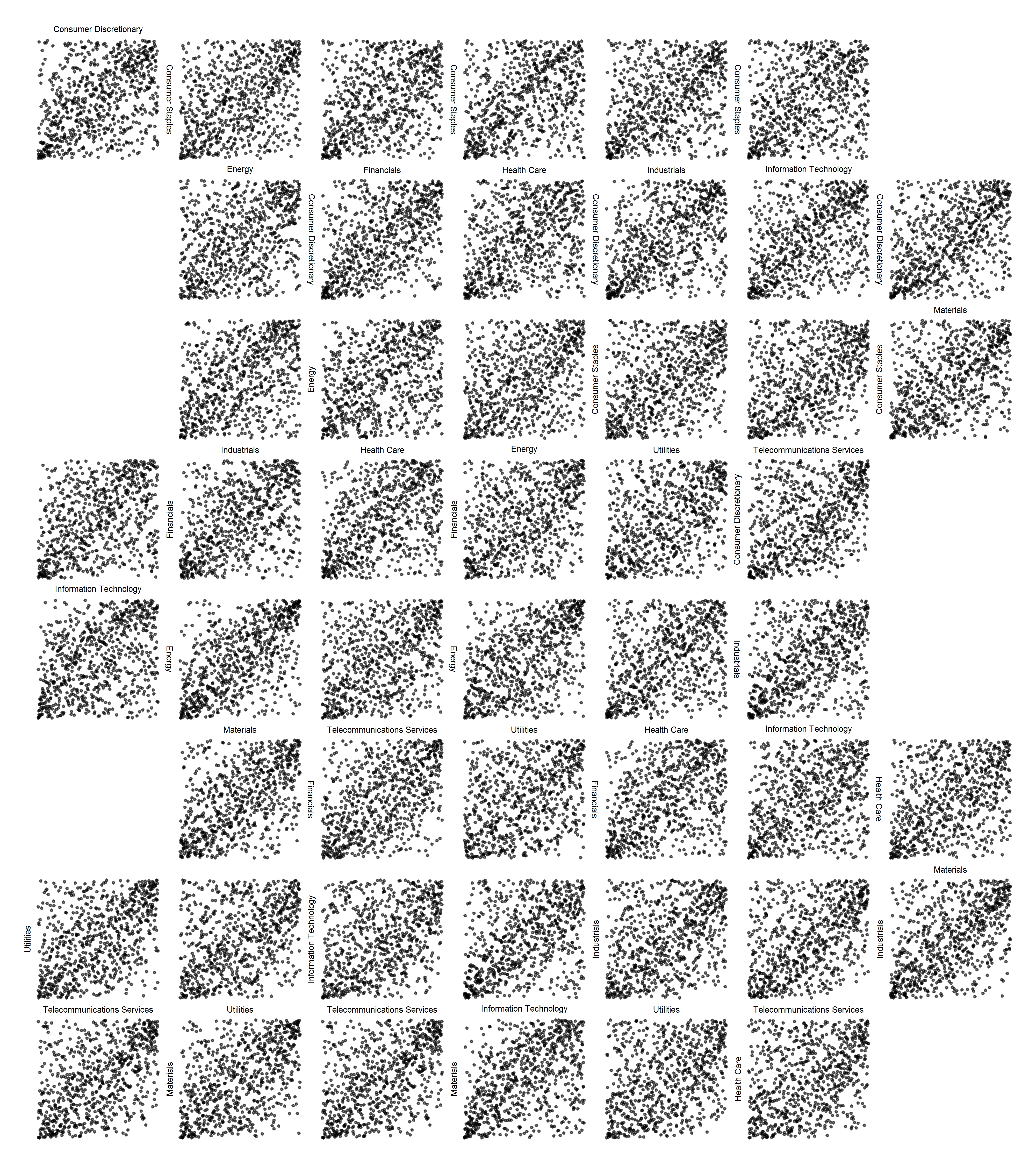}%
  \end{center}
  \caption{Zenplots displaying all pairs of pseudo-observations for the 10
    GICS sectors of the 465-dimensional S\&P~500 data based on the (Euclidean) distance (top
    left), weighted average (top right), PIT (bottom left), and maximum (bottom
    right) collapsing functions.}
  \label{fig:SP500:zenplot:independence:test:groups}
\end{figure}

\begin{figure}[htbp]
  \begin{center}
    \includegraphics[width=0.95\linewidth]{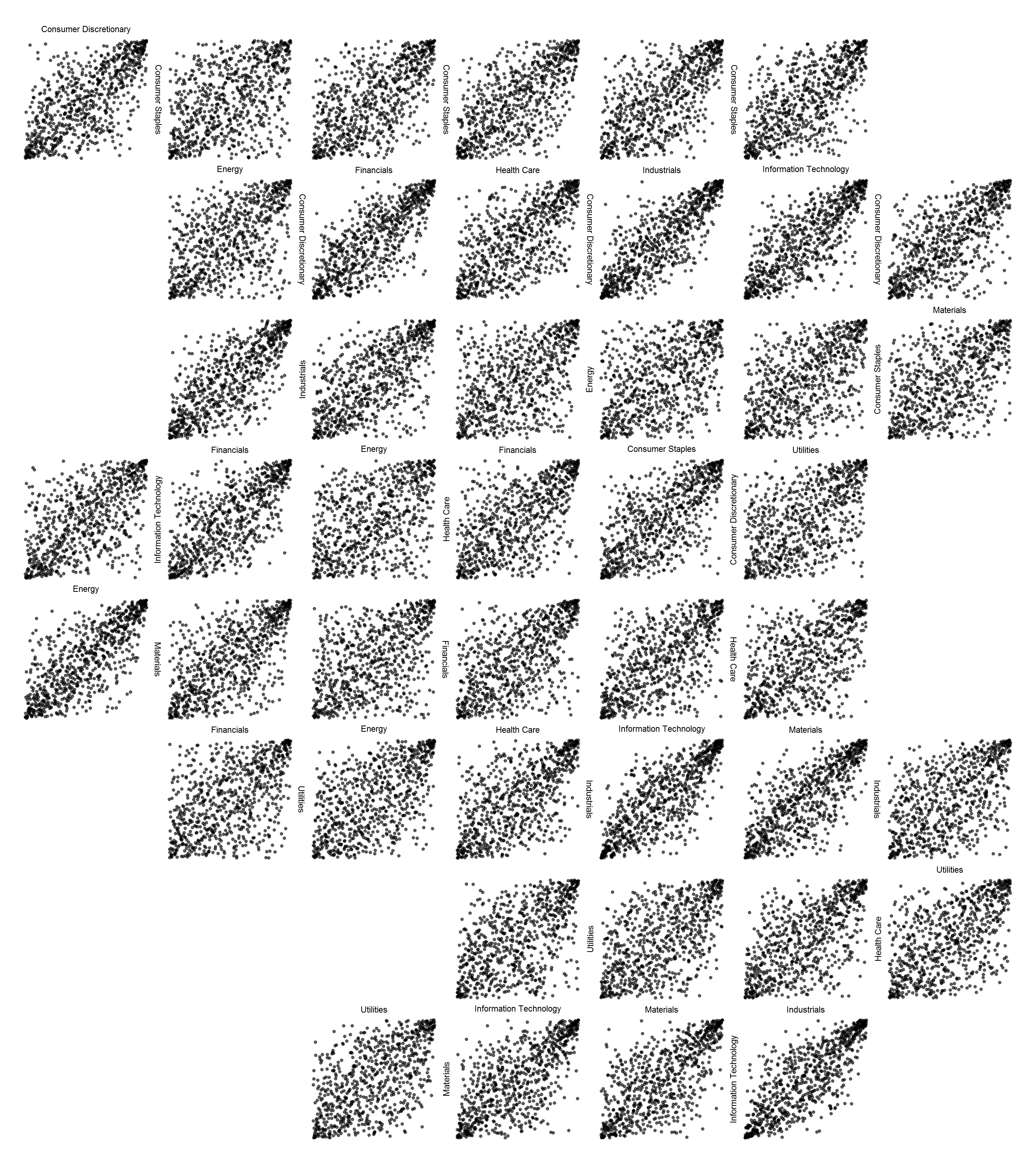}
  \end{center}
  \caption{Zenplot displaying all pairs of pseudo-observations for the nine GICS
    Sector ETFs.}
  \label{fig:SP500:zenplot:independence:test:etf}
\end{figure}

Figure~\ref{fig:SP500:zenplot:independence:test:groups} illustrates this
graphical assessment of independence with four zenplots, one for each choice of
collapsing function. As can be clearly detected from any of these choices of
collapsing functions, the business sectors cannot be assumed to be
independent. Furthermore to facilitate the comparison of the collapsed variables
with the benchmark, Figure~\ref{fig:SP500:zenplot:independence:test:etf}
depicts the pair-wise dependence structures between the nine sector ETFs.

The four plots in Figure~\ref{fig:SP500:zenplot:independence:test:groups} can be
interpreted as (data) realizations from the underlying (and unknown) collapsed
copula. Clearly realizations from the (Euclidean) distance collapsed copula is
denser in comparison to realizations from the other three collapsing functions
because it has $\binom{756}{2}$ realizations as opposed to just
756. In particular, due to the nature of the distance function, it is difficult
to interpret features (tail dependence, asymmetry, shape etc.) of the dependence
structure between business sectors in the context of the original variables
portrayed in the corresponding zenplot. As a result, for applications in
finance, the distance collapsing function should mostly be used for
(graphical) assessment of independence only.

Since the weighted average collapsing function is most natural for return data,
the interpretations of tail dependence and asymmetry translate well from the
bivariate case. We naturally see the similarity in the dependence
structures between the weighted average collapsing function and the benchmark (ETFs) in Figure~\ref{fig:SP500:zenplot:independence:test:etf}. Furthermore,
since the PIT collapsing function leads to realizations from the Kendall copula,
it also yields an attractive interpretation of the dependence structure depicted
in its corresponding zenplot. For instance as noted in Example~\ref{ex:tail_Kendallcop}, the tail dependence coefficients in
this case can be interpreted as natural multivariate extensions of bivariate
tail dependence. Owing to the justification of these two collapsing functions
and interpretability, one could potentially fit a copula model directly on the
collapsed variables if needed to model a notion of dependence between groups of
random variables, but this framework will in general not offer an analytically
tractable link back to the original random variables.

The maximum collapsing function appears to capture a weaker form of dependence
compared to the other collapsing functions and the benchmark. This is to be
expected as this collapsing function depicts a notion of dependence between the
worst performers only (in a plural sense in that the constituent chosen as the
maximum can change daily in each business sector over the time period
considered). One would use this collapsing function only if one is interested in
such a notion of dependence between random vectors.

\subsubsection{Dynamic S\&P~500 sector dependence}
Having garnered an understanding of the dependence structures between business
sectors in a fixed time period, we will now dynamically capture the dependence
between these sectors using a moving window setup. In particular, we investigate
time varying dependence from 2006-01-01 to 2015-12-31. Due to our missing data
handling, we will be working with a subset of 461 time series from the
S\&P~500 data.

Figure~\ref{fig: time-varying_collapsedvsetf} depicts for four randomly chosen
pairs of business sectors, the time-varying dependence as captured by the
distance, average, maximum, and PIT collapsing functions. Also included (for a
form of comparison with the benchmark) is the dependence measure between ETFs
for each pair. Note that we used a 250-day moving window for the plots in Figure~\ref{fig: time-varying_collapsedvsetf}. The first takeaway
from the four collapsed and ETF dependence series is that they seem to capture a
very similar shape across time. While the dependence measures resulting from
different collapsing functions lie on different scales, they all capture the
same shifts in dependence not only with respect to each other but also with
respect to the market-determined ETF dependence series. This indicates the
suitability of any of these collapsing function in the task of detecting
dependence and the shifts in the strength of dependence over time. Furthermore, note that ETFs are marketed as weighted average of sector constituents, but are
tradeable securities themselves and thus exposed to market forces. Such a
construction of the ETFs explains why the average collapsing function would most
closely track the dependence between ETFs (despite the use of equal weights in
our collapsing function and the influence that market forces might have on the
dependence between sector ETFs).

\begin{figure}[htbp]
  \begin{center}
    \includegraphics[width=0.49\linewidth]{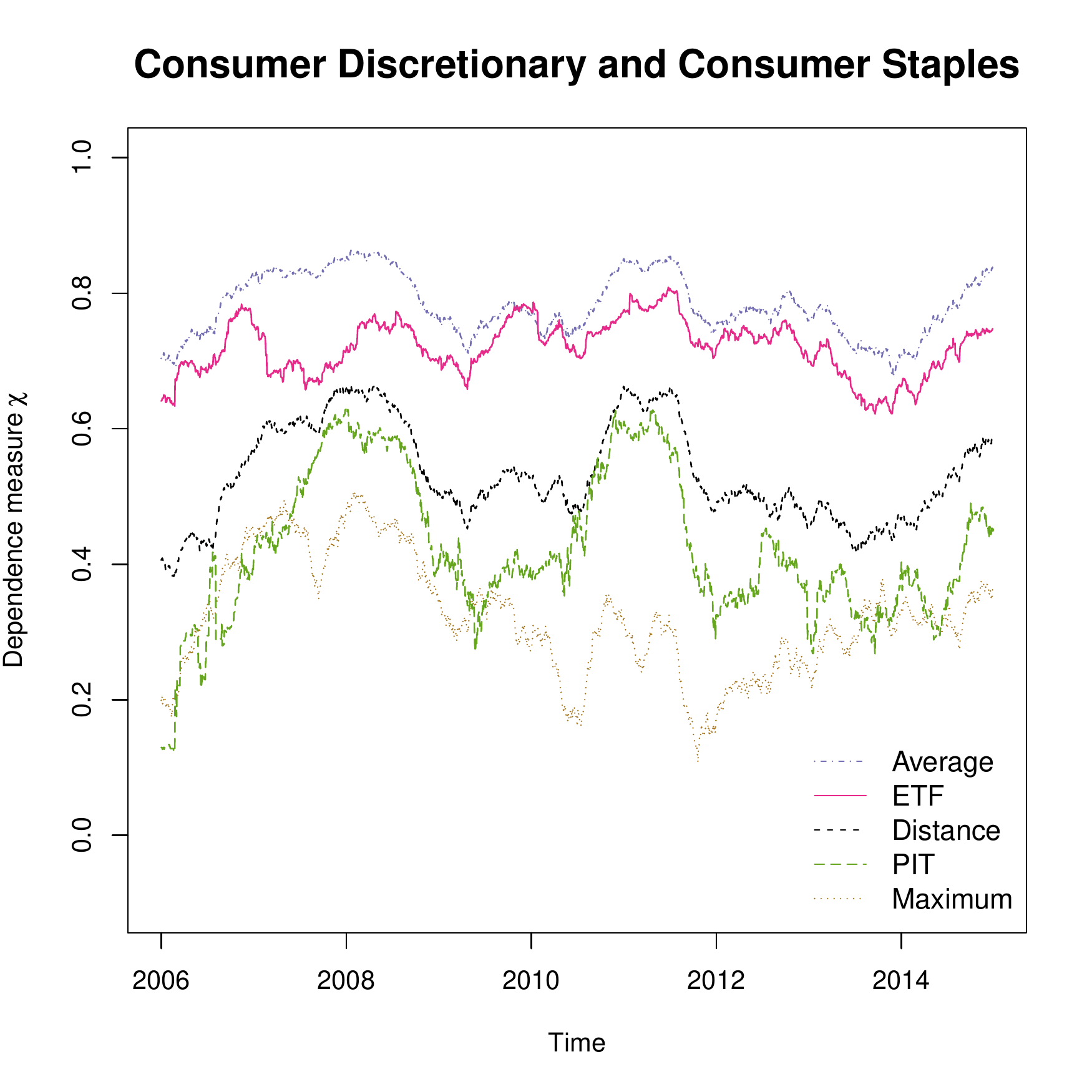}%
    \hspace{1mm}
    \includegraphics[width=0.49\linewidth]{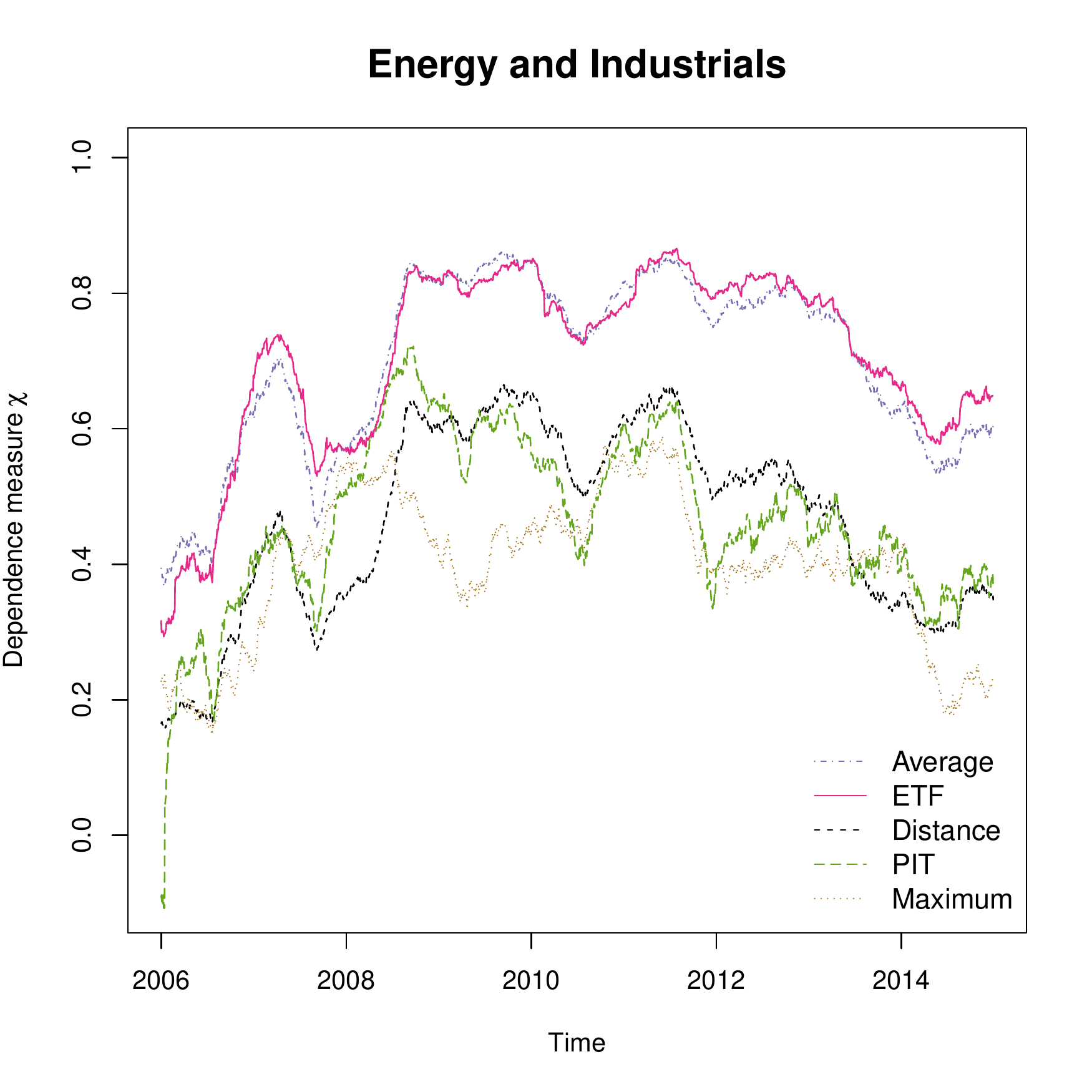}

    \vspace{1mm}
    \includegraphics[width=0.49\linewidth]{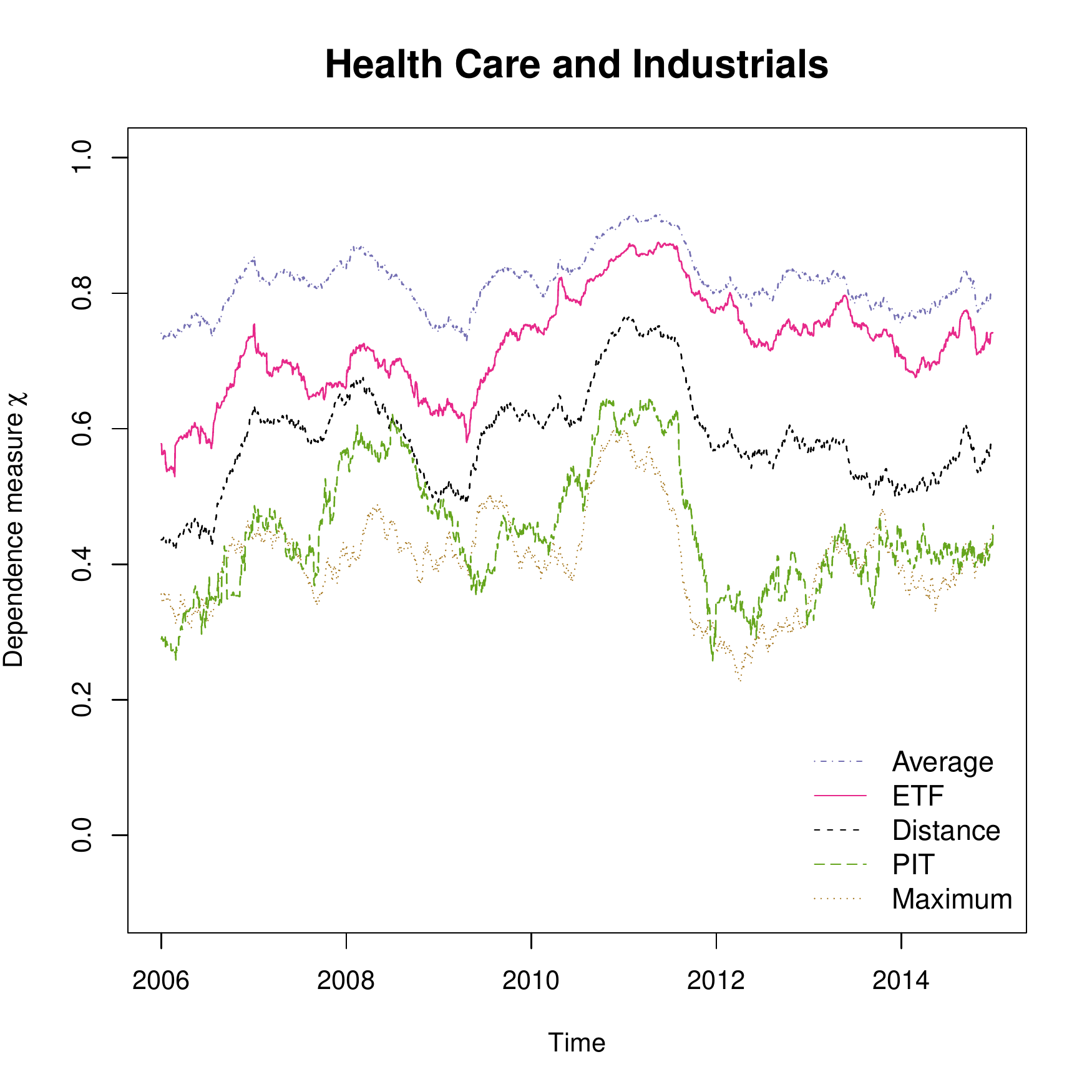}%
    \hspace{1mm}
    \includegraphics[width=0.49\linewidth]{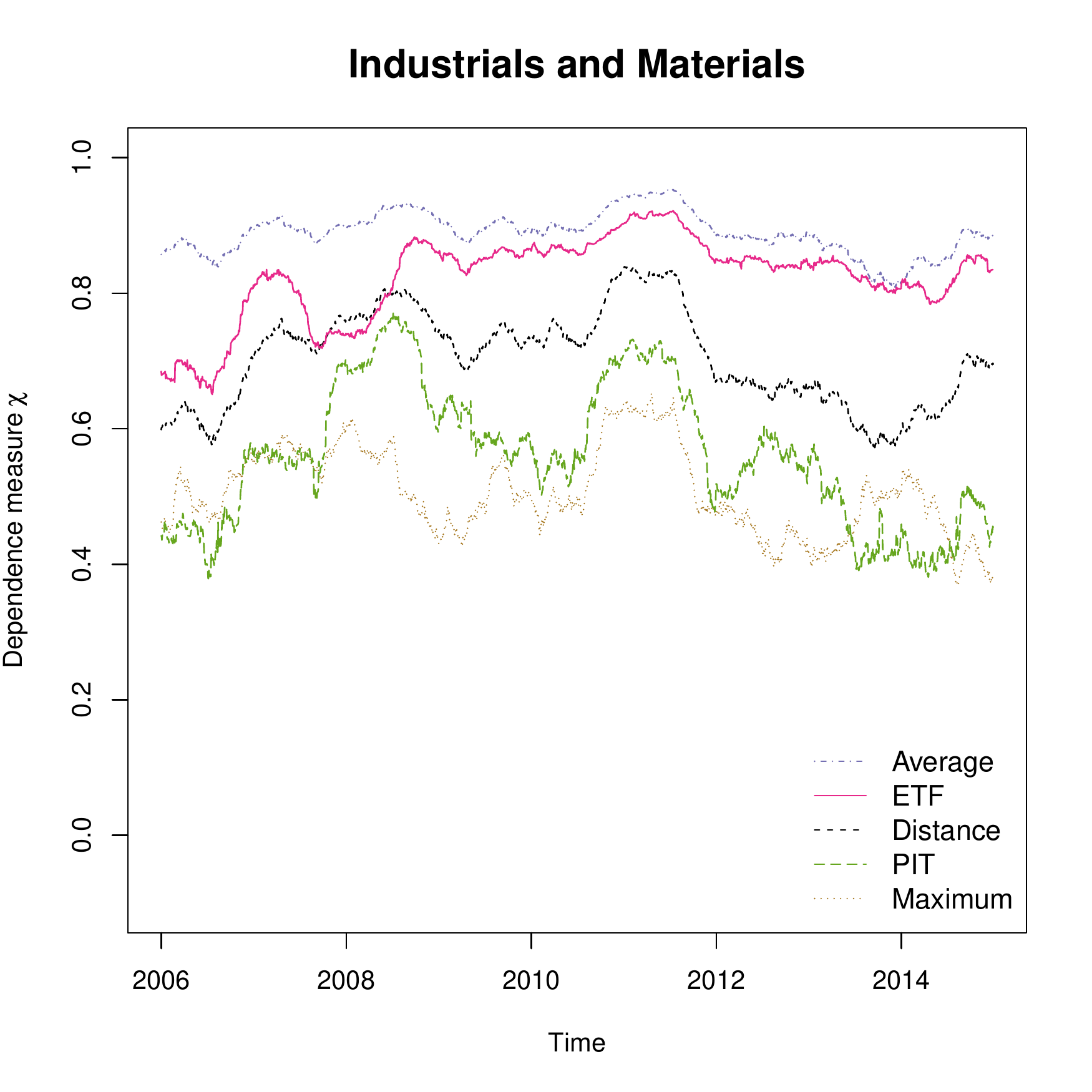}%
  \end{center}
  \caption{Time-varying dependence measure for various collapsing functions and
    the ETFs between a few selected pairs of business sectors. The
    four pairs of sectors arbitrarily selected are as follows: Consumer
    Discretionary vs. Consumer Staples (top left), Energy vs. Industrials (Top
    right), Health Care vs. Industrials (bottom left), and Industrials vs
    Materials (bottom right).}
  \label{fig: time-varying_collapsedvsetf}
\end{figure}

Figure~\ref{fig: time-varying_collapsedvsbiv} showcases the time-varying
dependence as captured by the distance, average, and maximum collapsing
functions with their corresponding confidence intervals constructed using
Proposition~\ref{proof:Prop 4.1} and Remark~\ref{remark:est}. For the plots in
this figure, we used a 150-day moving window. Shown in the
background are all pair-wise (bivariate) time-varying dependence
measures between individual constituents of the two sectors. This juxtaposition
highlights that the dependence measures between collapsed random variables
capture fairly similar shifts in strength of dependence over time compared with
all the pair-wise (classical) dependence measures between the
sectors. Furthermore, one can see that the width of confidence intervals for the
various collapsed dependence measures is well-within the width of the background
band representing all the bivariate dependence series between individual
constituents from each sector. This further provides some intuitive
corroboration that the collapsing functions in some sense sufficiently capture
time-varying dependence between groups of random variables (that is, sufficient
when compared to a series of matrices of pair-wise dependence measures).
\begin{figure}[htbp]
  \begin{center}
    \includegraphics[width=0.48\linewidth]{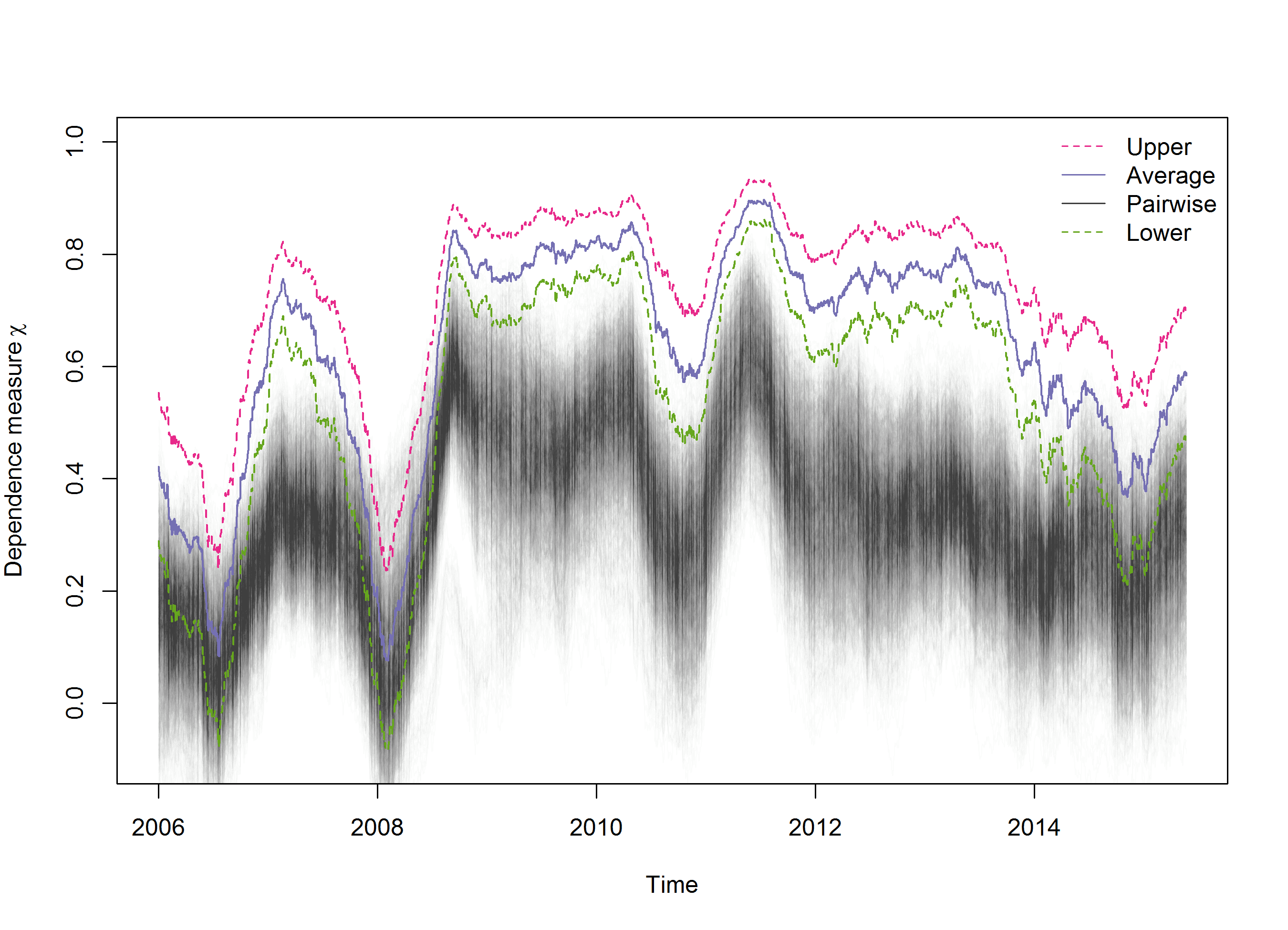}
    \hspace{1mm}
    \includegraphics[width=0.48\linewidth]{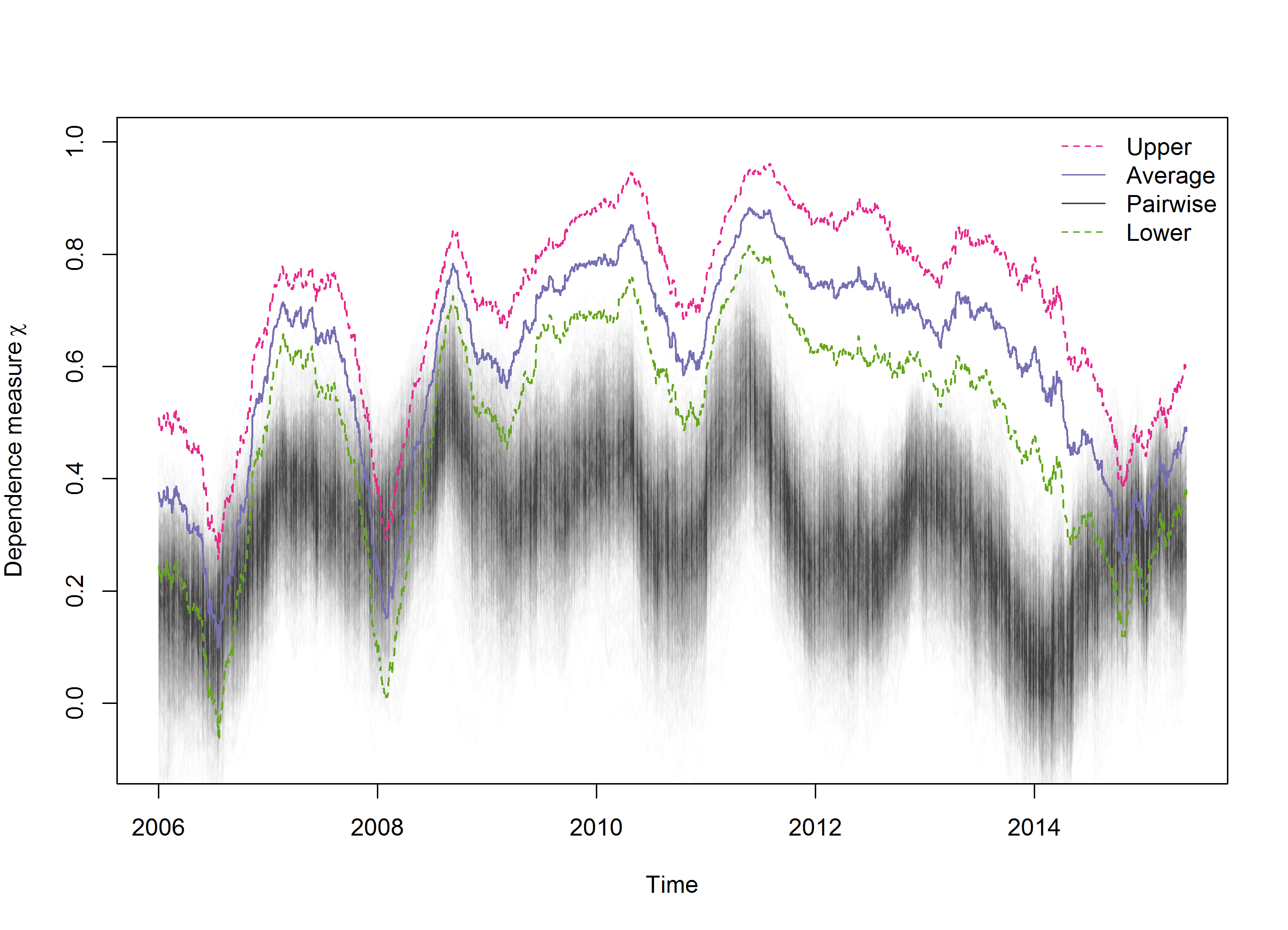}
    \vspace{0.1mm}
    \includegraphics[width=0.48\linewidth]{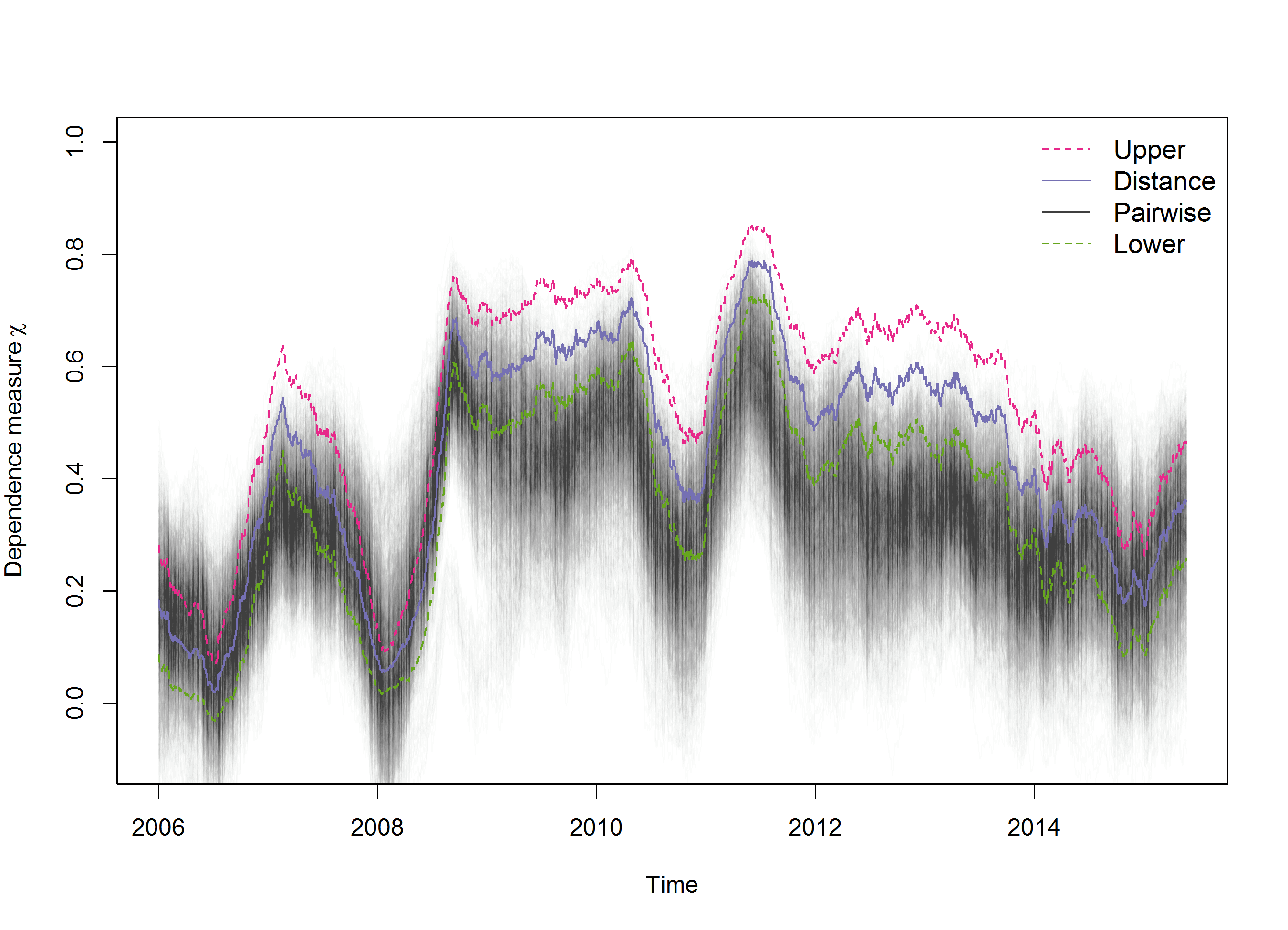}%
    \hspace{1mm}
    \includegraphics[width=0.48\linewidth]{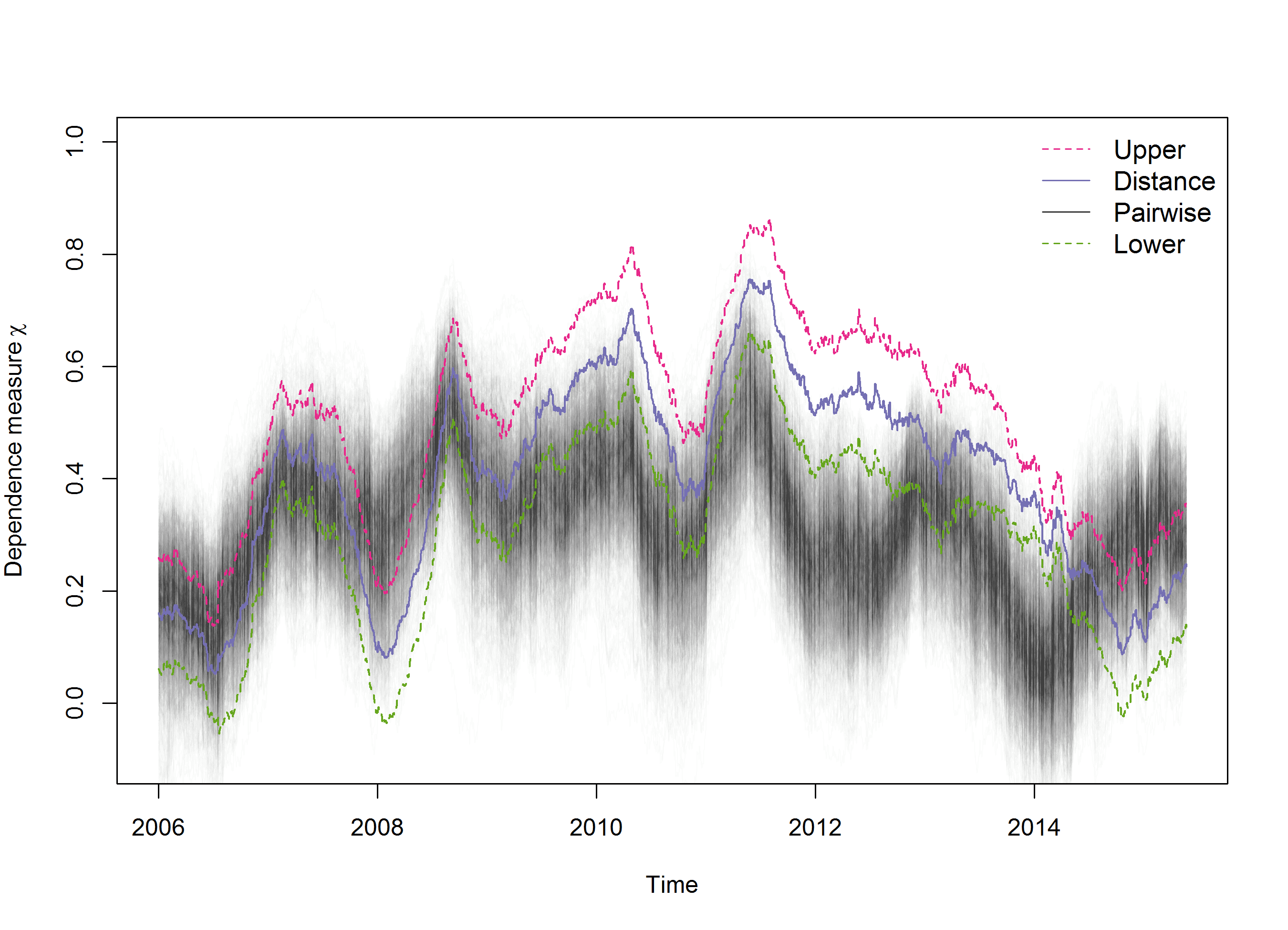}%
    \vspace{0.1mm}
    \includegraphics[width=0.48\linewidth]{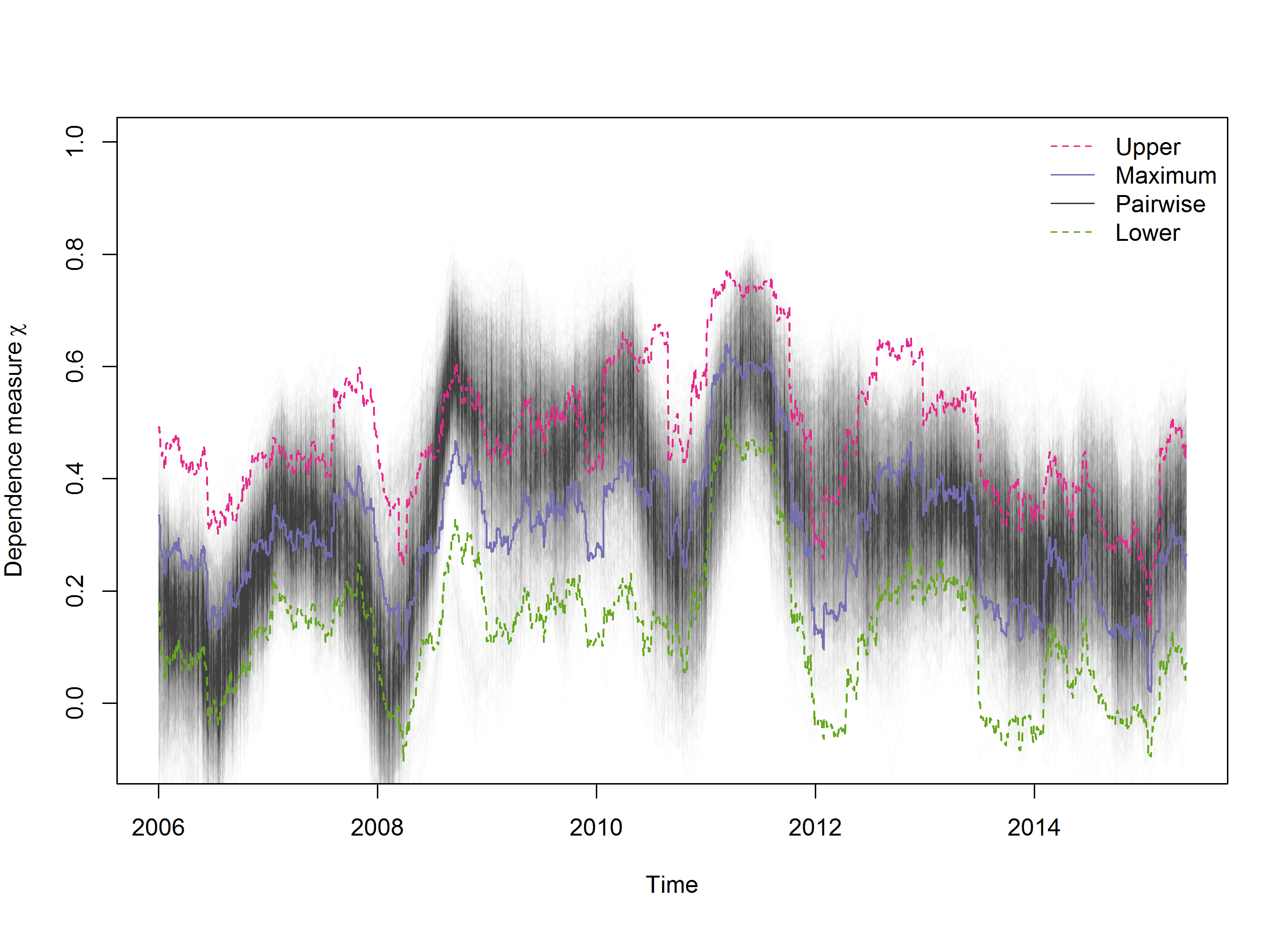}%
    \hspace{1mm}
    \includegraphics[width=0.48\linewidth]{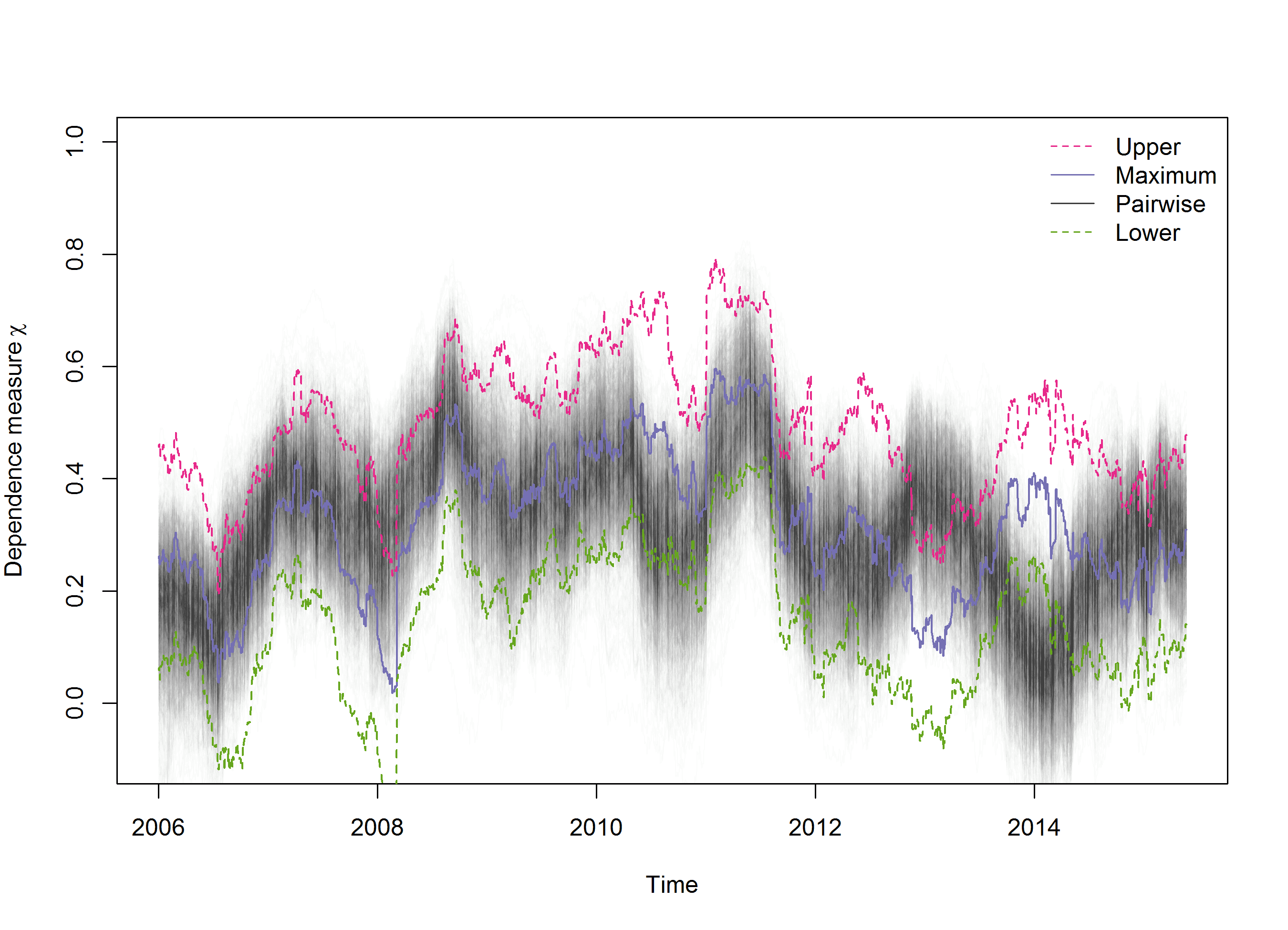}%
  \end{center}
  \caption{Time-varying dependence measure for average (top plots),
    distance (middle plots), and maximum (bottom plots) collapsing functions with 95\%
    confidence intervals against a backdrop of all pair-wise
    time-varying measures between assets in two business sectors. On the left
    panel we present the plots for Consumer Discretionary vs. Energy sectors and
    on the right panel we present the plots for Energy vs. Health Care
    sectors.}
  \label{fig: time-varying_collapsedvsbiv}
\end{figure}

\section{Conclusion and discussion}
In this paper, we introduced a framework for quantifying dependence between
random vectors. With the notion of a collapsing function, random vectors were
summarized by collapsed random variables which were then used as a proxy for the
purposes of studying dependence between random vectors. Based on this framework,
a graphical assessment of independence between random vectors was proposed and
its applicability demonstrated with examples from finance. Furthermore various
measures of association and dependence between random vectors were suggested as
a natural by-product of this framework.

Additionally, we introduced and explored the notion of a collapsed distribution
function and collapsed copula for the maximum and PIT collapsing functions. As a
result, we were able to relate the dependence between collapsed random variables
to the dependence between the original random vectors. Particularly for both of
these collapsing functions, we derived analytical forms of the collapsed
distribution or copula in the Archimedean case. Moreover, for the PIT collapsing
function, this lead to a multivariate extension of Kendall distributions.

General asymptotic results for the dependence measures resulting from the
framework were derived through the lens of U-statistics with the exception of
the PIT collapsing function for which an estimator not amenable to the
U-statistics theory was proposed. These asymptotic results then allowed us to
construct confidence intervals for collapsed dependence measures constructed
within our framework. Our results were showcased in Section~\ref{sec:ex:SP500}
where we captured the evolution of dependence between business sectors over
time. In addition, for this finance example we visualized dependence between
sectors via realizations from various collapsed copulas. Beyond this example, we
considered protein data from the realm of bioinformatics. The task involved to
rank pairs of residues which were modeled by random vectors of varying
dimensions. Dependence measures resulting from our framework were thus naturally
used as a metric for this ranking task. We showed that for some collapsing
functions, our measures were fairly comparable with prior tailor-made methods
used in the literature while requiring a lesser computational burden.

We conclude that there is no notion of a ``best'' collapsing function. All
reasonable collapsing functions we investigated tend to capture the dependence
between random vectors in a similar fashion with subtle variations. However,
there are some advantages for each collapsing function that are noteworthy. The
maximum and PIT collapsing functions allow for a (more direct) link between the
collapsed and the original (high-dimensional) distribution function or
copula. Furthermore, the PIT and multivariate rank collapsing functions lead to
multivariate extensions of Kendall's tau and Spearman's rho as highlighted in
\cite{grothe2014}. The distance function is often a good choice of collapsing
function for the graphical assessment of independence between random
vectors. Furthermore, the distance and kernel collapsing functions, if carefully
chosen, can naturally detect various non-linear (beyond monotone) dependencies
between groups of random variables. The maximum and weighted average collapsing
functions require the least computational time. Moreover, the weighted average
collapsing function seems natural in the context of finance and yields very
competitive results for the ranking task in the protein example.

An interesting and open challenge for our collapsing function framework lies in
understanding the relationship between the collapsed copula and the inherent
higher dimensional copula between the original random variables. Naturally, one
loses information when compressing random vectors into single random
variables. Having an explicit connection between the collapsed copula and the
original copula helps in better understanding this loss of information. Some
collapsing functions such as distance or kernel functions involve complicated
non-linear transformations of the original random variables and hence render
this task complicated. For the weighted average collapsing function,
understanding the collapsed copulas remains a pertinent and open question.

\printbibliography[heading=bibintoc]

\appendix

\section{Proofs and additional details for the asymptotic framework}
\subsection{Proof of Proposition~\ref{Prop:Asymptotics}}\label{proof:Prop 4.1}
\begin{proof}
  We begin by explicitly writing out the population version of our dependence
  measure. For a general collapsing function $S$,
  \begin{align*}
    &\chi(\bm{X},\bm{Y})=\rho(S(\bm{X}),S(\bm{Y}))=\frac {\mu_{xy}-\mu_x\mu_y}
      {\sqrt{\mu_{xx}-\mu_x^2} \sqrt{\mu_{yy}-\mu_y^2} }.
  \end{align*}
  \subsubsection*{Case 1: $S$ is a $p$-variate function}
  Based on the $n$ independent random samples, define
  \begin{align*}
    &m^{(1)}_{x}=\frac{1}{n} \sum_{i=1}^{n} S(\bm{X}_i),\quad  m^{(1)}_y=\frac{1}{n} \sum_{i=1}^{n} S(\bm{Y}_i),\quad  m^{(1)}_{xx}=\frac {1} {n} \sum_{i=1}^{n} S(\bm{X}_i)^2\\
    & m^{(1)}_{yy}=\frac {1} {n} \sum_{i=1}^{n} S(\bm{Y}_i)^2,\quad
      m^{(1)}_{xy}= \frac {1} {n} \sum_{i=1}^{n} S(\bm{X}_i)  S(\bm{Y}_i).
  \end{align*}
  By \cite{hoeffding1948}, $m^{(1)}_x$, $m^{(1)}_y$, $m^{(1)}_{xx}$,
  $m^{(1)}_{yy}$, $m^{(1)}_{xy}$ are U-statistics for $\mu_x$, $\mu_y$,
  $\mu_{xx}$, $\mu_{yy}$, $\mu_{xy}$ respectively. Following from Hoeffding's
  decomposition theorem, see \cite[Chapter~3]{lee1990}, we can
  conclude that, as $n\to\infty$,
  \begin{align*}
    &\sqrt{n}(m^{(1)}_x-\mu_x) = \frac{1}{\sqrt{n}} \sum_{i=1}^{n} (S(\bm{X}_i)-\mu_x) + o_{\text{p}}(1),\\
    &\sqrt{n}(m^{(1)}_y-\mu_y) = \frac{1}{\sqrt{n}} \sum_{i=1}^{n} (S(\bm{Y}_i)-\mu_y) + o_{\text{p}}(1),\\
    &\sqrt{n}(m^{(1)}_{xx}-\mu_{xx}) = \frac{1}{\sqrt{n}} \sum_{i=1}^{n} (S(\bm{X}_i)^2-\mu_{xx}) + o_{\text{p}}(1),\\
    &\sqrt{n}(m^{(1)}_{yy}-\mu_{yy}) = \frac{1}{\sqrt{n}} \sum_{i=1}^{n} (S(\bm{Y}_i)^2-\mu_{yy}) + o_{\text{p}}(1),\\
    &\sqrt{n}(m^{(1)}_{xy}-\mu_{xy}) = \frac{1}{\sqrt{n}} \sum_{i=1}^{n} (S(\bm{X}_i)S(\bm{Y}_i)-\mu_{xy}) + o_{\text{p}}(1).
  \end{align*}
  Combining all the terms, it follows that, for $n\to\infty$,
  \begin{align*}
    \sqrt{n} \begin{pmatrix}
      m^{(1)}_x -\mu_x\\ m^{(1)}_y-\mu_y \\m^{(1)}_{xx}-\mu_{xx} \\ m^{(1)}_{yy}-\mu_{yy}
      \\ m^{(1)}_{xy}-\mu_{xy}
    \end{pmatrix} \overset{\text{d}}{\longrightarrow} \N_5 (\bm{0},\Sigma_1),
  \end{align*}
  where $\Sigma_1$ is the covariance matrix of the random vector
  \begin{align*}
    (S(\bm{X}),S(\bm{Y}),S(\bm{X})^{2},S(\bm{Y})^{2},S(\bm{X})S(\bm{Y})).
  \end{align*}
  Then, we construct an estimator for the population dependence measure, $\chi(\bm{X},\bm{Y})$, as a function of the U-statistics.
  \begin{align*}
    \chi_n(\bm{X},\bm{Y})=
    f(m^{(1)}_x,m^{(1)}_y,m^{(1)}_{xx},m^{(1)}_{yy},m^{(1)}_{xy})=\frac {m^{(1)}_{xy}-m^{(1)}_xm^{(1)}_y}
    {\sqrt{m^{(1)}_{xx}-(m^{(1)}_x)^2} \sqrt{m^{(1)}_{yy}-(m^{(1)}_y})^2 },
  \end{align*}
  where $m^{(1)}_x$, $m^{(1)}_y$, $m^{(1)}_{xx}$, $m^{(1)}_{yy}$, and $m^{(1)}_{xy}$ are the sample quantities as previously defined.
  Then, by the delta method we have
  \begin{align*}
    \sqrt{n} (\chi_n(\bm{X},\bm{Y})-\chi(\bm{X},\bm{Y}))\overset {d} {\longrightarrow} \N(0,\sigma^2_{\chi}), \quad (n\rightarrow \infty),
  \end{align*}
  where  $\sigma_{\chi}^2=(\nabla f_{5 \times 1}|_{\bm{\mu}})'\Sigma_1(\nabla f_{5 \times 1}|_{\bm{\mu}})$. Note that the gradient vector is evaluated at $\bm{\mu}=(\mu_x,\mu_y,\mu_{xx},\mu_{yy},\mu_{xy})$.

  \subsubsection*{Case 2: $S$ is a $2p$-variate function}
  Consider
  \begin{align*}
    &m^{(2)}_{x}=\frac{1}{{n \choose 2}} \sum_{i=1}^{n}\sum_{j>i}^{n} S(\bm{X}_i,\bm{X}_j),\  m^{(2)}_y=\frac{1}{{n \choose 2}} \sum_{i=1}^{n}\sum_{j>i}^{n} S(\bm{Y}_i,\bm{Y}_j),\  m^{(2)}_{xx}=\frac {1} {{n \choose 2}} \sum_{i=1}^{n}\sum_{j>i}^{n} S(\bm{X}_i,\bm{X}_j)^2\\
    & m^{(2)}_{yy}=\frac {1} {{n \choose 2}} \sum_{i=1}^{n}\sum_{j>i}^{n} S(\bm{Y}_i,\bm{Y}_j)^2,\
      m^{(2)}_{xy}= \frac {1} {{n \choose 2}} \sum_{i=1}^{n}\sum_{j>i}^{n} S(\bm{X}_i,\bm{X}_j)  S(\bm{Y}_i,\bm{Y}_j).
  \end{align*}
  Similar to the setup presented in Case 1, these sample quantities are
  naturally U-statistics for their corresponding population quantities.  Again
  following from Hoeffding's decomposition theorem, we have that, as $n\to\infty$,
  \begin{align*}
    &\sqrt{n}(m^{(2)}_x-\mu_x) = \frac{2}{\sqrt{n}} \sum_{i=1}^{n} \big(\E_{\tiny \bm{X}'}\big[S(\bm{X}_i,\bm{X}')\big]-\mu_x\big) + o_{\text{p}}(1),\\
    &\sqrt{n}(m^{(2)}_y-\mu_y) = \frac{2}{\sqrt{n}} \sum_{i=1}^{n} (\E_{\tiny \bm{Y}'}\big[S(\bm{Y}_i,\bm{Y}')\big]-\mu_y) + o_{\text{p}}(1),\\
    &\sqrt{n}(m^{(2)}_{xx}-\mu_{xx}) = \frac{2}{\sqrt{n}} \sum_{i=1}^{n} (\E_{\tiny \bm{X}'}\big[S(\bm{X}_i,\bm{X}')^2 \big]-\mu_{xx}) + o_{\text{p}}(1),\\
    &\sqrt{n}(m^{(2)}_{yy}-\mu_{yy}) = \frac{2}{\sqrt{n}} \sum_{i=1}^{n} (\E_{\tiny \bm{Y}'}\big[S(\bm{Y}_i,\bm{Y}')^2\big]-\mu_{yy}) + o_{\text{p}}(1),\\
    &\sqrt{n}(m^{(2)}_{xy}-\mu_{xy}) = \frac{2}{\sqrt{n}} \sum_{i=1}^{n} (\E_{\tiny (\bm{X}',\bm{Y}')}\big[S(\bm{X}_i,\bm{X}')S(\bm{Y}_i,\bm{Y}')\big]-\mu_{xy}) + o_{\text{p}}(1),
  \end{align*}
  where the conditional expectations in the expressions above represent the
  first order Hoeffding decomposition of the corresponding U-statistic.
  Combining all the terms, it follows that
  \begin{align*}
    \sqrt{n} \begin{pmatrix}
      m^{(2)}_x -\mu_x\\ m^{(2)}_y-\mu_y \\m^{(2)}_{xx}-\mu_{xx} \\ m^{(2)}_{yy}-\mu_{yy}
      \\ m^{(2)}_{xy}-\mu_{xy}
    \end{pmatrix} \overset{\text{d}}{\longrightarrow} \N_5 (\bm{0},4 \Sigma_2),
  \end{align*}
  where $\Sigma_2$ is the covariance matrix of the random vector
  \begin{align*}
    (\E_{\tiny \bm{X}'}\big[S(\bm{X},\bm{X}')\big],\E_{\tiny \bm{Y}'}\big[S(\bm{Y},\bm{Y}')\big],\E_{\tiny \bm{X}'}\big[S(\bm{X},\bm{X}')^2\big],\E_{\tiny \bm{Y}'}\big[S(\bm{Y},\bm{Y}')^2\big],\E_{\tiny (\bm{X}',\bm{Y}')}\big[S(\bm{X},\bm{X}')S(\bm{Y},\bm{Y}')\big]).
  \end{align*}
  We can then construct an estimator for the population dependence measure
  $\chi(\bm{X},\bm{Y})$ exactly as we did in
  Case~1 but instead with the use of the sample quantities $m^{(2)}_x$,
  $m^{(2)}_y$, $m^{(2)}_{xx}$, $m^{(2)}_{yy}$, $m^{(2)}_{xy}$. By
  the delta method we have that, as $n\to\infty$,
  \begin{align*}
    \sqrt{n} (\chi_n(\bm{X},\bm{Y})-\chi(\bm{X},\bm{Y}))\overset {d} {\longrightarrow} \N(0,\sigma^2_{\chi}),
  \end{align*}
  where $\sigma_{\chi}^2=(\nabla f_{5 \times 1}|_{\bm{\mu}})'\Sigma_2(\nabla f_{5 \times 1}|_{\bm{\mu}})$.
\end{proof}

\subsection{Additional details for estimation of
  $\sigma_{\chi}^2$} \label{App:add_details}
Analytical forms of the components of the gradient vector are given below; note
that $\bm{m}=(m_x,m_y,m_{xx},m_{yy},m_{xy})$ acts as a place holder for both
$\bm{m}^{(1)}$ and $\bm{m}^{(2)}$ defined in Remark \ref{remark:est}:
\begin{align*}
  &\nabla f_1\rvert_{\bm{m}}= \frac{m_x(m_{xy}-m_{x}m_{y})}{(m_{xx}-m_x^2)^{3/2}\sqrt{m_{yy}-m_y^2}} -\frac{m_y}{\sqrt{m_{xx}-m_x^2} \sqrt{m_{yy}-m_y^2}}, \\
  &\nabla f_2\rvert_{\bm{m}}= \frac{m_y(m_{xy}-m_{x}m_{y})}{(m_{yy}-m_y^2)^{3/2}\sqrt{m_{xx}-m_x^2}} -\frac{m_x}{\sqrt{m_{xx}-m_x^2} \sqrt{m_{yy}-m_y^2}}, \\
  &\nabla f_3\rvert_{\bm{m}}=
    -\frac{m_{xy}-m_{x}m_{y}}{2(m_{xx}-m_x^2)^{3/2} \sqrt{m_{yy}-m_y^2}},\\
  &\nabla f_4\rvert_{\bm{m}}=
    -\frac{m_{xy}-m_{x}m_{y}}{2(m_{yy}-m_y^2)^{3/2} \sqrt{m_{xx}-m_x^2}},\\
  &\nabla f_5\rvert_{\bm{m}}=\frac{1}{\sqrt{m_{xx}-m_x^2} \sqrt{m_{yy}-m_y^2}}.
\end{align*}

\subsection{Additional asymptotic results}
An estimator $\tau_n$ of
$\tau(S(\bm X), S(\bm Y))=\rho(\I_{\{S(\bm X)\le S(\bm X')\}}, \I_{\{S(\bm Y)\le S(\bm
	Y')\}})$ can be constructed through the U-statistics framework with the corresponding asymptotic results following as a consequence of Proposition \ref{Prop:Asymptotics}.
\begin{corollary}[Asymptotic distribution of $\tau_n$]
Assume $\bm{X}',\bm{X}'',\bm{X}'''$ are independent copies of $\bm{X}$. Suppose $\tau_n(\bm{X},\bm{Y})$ is constructed as a function of U-statistics. Then, as $n\to\infty$,
\begin{align*}
\sqrt{n} ({\tau}_n(S(\bm{X}),S(\bm{Y}))-\tau(S(\bm{X}),S(\bm{Y})))\overset {d} {\longrightarrow} \N(0,\sigma^2_{\tau}),
\end{align*}
where
\begin{align*}
\sigma_{\tau}^2=\begin{cases}
4(\nabla f_{3 \times 1}|_{\bm{\mu}})'\Sigma_1(\nabla f_{3 \times 1}|_{\bm{\mu}}),&\text{if $S$ is a $p$-variate function},\\
16(\nabla f_{3 \times 1}|_{\bm{\mu}})'\Sigma_2(\nabla f_{3 \times 1}|_{\bm{\mu}}),&\text{if $S$ is a $2p$-variate function}.
\end{cases}
\end{align*}
Here, $\nabla f_{3 \times 1}|_{\bm{\mu}}$ denotes the gradient vector
of the function
\begin{align*}
f(a,b,c)=\frac {c-ab} {\sqrt{a-a^2} \sqrt{b-b^2}},
\end{align*}
evaluated at the population mean
$\bm{\mu}=(\mu_x,\mu_y,\mu_{xy})$, where
$\mu_x=\P[S(\bm{X}) \leq S(\bm{X}')]$, $\mu_y=\P[S(\bm{Y}) \leq S(\bm{Y}')]$ and $\mu_{xy}=\P[S(\bm{X})\leq S(\bm{X}'),S(\bm{Y})\leq S(\bm{Y}')]$.
Furthermore, $\Sigma_1$ denotes the covariance matrix of
\begin{align*}
(\P_{\bm{X}'|\bm{X}}[S(\bm{X}) \leq S(\bm{X}')],\P_{\bm{Y}'|\bm{Y}}[S(\bm{Y}) \leq S(\bm{Y}')],\P_{\bm{X}',\bm{Y}'|\bm{X},\bm{Y}}[S(\bm{X}) \leq S(\bm{X}'), S(\bm{Y}) \leq S(\bm{Y}')])
\end{align*}
and $\Sigma_2$ denotes the covariance matrix of
\begin{align}
&(\P_{\bm{X}',\bm{X}'',\bm{X}'''|\bm{X}}[S(\bm{X},\bm{X}') \leq S(\bm{X}'',\bm{X}''')], \P_{\bm{Y}',\bm{Y}'',\bm{Y}'''|\bm{Y}}[S(\bm{Y},\bm{Y}') \leq S(\bm{Y}'',\bm{Y}''')],\notag\\
&\P_{\bm{X}',\bm{X}'',\bm{X}''',\bm{Y}',\bm{Y}'',\bm{Y}'''|\bm{X},\bm{Y}}[S(\bm{X},\bm{X}') \leq S(\bm{X}'',\bm{X}'''),S(\bm{Y},\bm{Y}') \leq S(\bm{Y}'',\bm{Y}''')]),  \label{tau:case2:randomvec}
\end{align}
where $\P_{\cdot|\cdot}$ denotes a conditional probability.
\end{corollary}
\begin{proof}
 We begin by explicitly writing out the population version of our dependence
 measure. For a general collapsing function $S$,
 \begin{align*}
 &\tau(S(\bm{X}),S(\bm{Y}))=\rho(\I_{\{S(\bm X)\le S(\bm X')\}}, \I_{\{S(\bm Y)\le S(\bm{Y}')\}})=\frac {\mu_{xy}-\mu_x\mu_y}
 {\sqrt{\mu_{x}-\mu_x^2} \sqrt{\mu_{y}-\mu_y^2} }.
 \end{align*}
  \subsubsection*{Case 1: $S$ is a $p$-variate function}
  Based on a random sample $(\bm{X}_1,\bm{Y}_1),\dots,(\bm{X}_n,\bm{Y}_n)$,
  estimators $m_x^{(1)},m_y^{(1)},$ and $m_{xy}^{(1)}$ can be constructed using
  the setup of the proof of Case~2 of Proposition~\ref{Prop:Asymptotics}. The convergence result follows from
  a similar delta method argument.
  \subsubsection*{Case 2: $S$ is $2p$-variate function}
    The sample quantities
    \begin{gather*}
    m^{(2)}_{x}=\frac{1}{{n \choose 4}} \sum_{i<j<k<l} \I_{\{S(\bm{X}_i,\bm{X}_j')\leq S(\bm{X}_k'',\bm{X}_l''')\}},\quad  m^{(2)}_{y}=\frac{1}{{n \choose 4}} \sum_{i<j<k<l} \I_{\{S(\bm{Y}_i,\bm{Y}_j')\leq S(\bm{Y}_k'',\bm{Y}_l''')\}},\ \\
    m^{(2)}_{xy}=\frac{1}{{n \choose 4}} \sum_{i<j<k<l} \I_{\{S(\bm{X}_i,\bm{X}_j')\leq S(\bm{X}_k'',\bm{X}_l'''),S(\bm{Y}_i,\bm{Y}_j')\leq S(\bm{Y}_k'',\bm{Y}_l''')\}}
    \end{gather*}
    are naturally U-statistics for their corresponding population
    quantities. Then, following Hoeffding's decomposition theorem, we have that,
    as $n \to \infty$,
\begin{align*}
\sqrt{n}(m^{(2)}_x-\mu_x) &= \frac{4}{\sqrt{n}} \sum_{i=1}^{n} \bigl(\P_{\tiny \bm{X}',\bm{X}'',\bm{X}'''|\bm{X}}(S(\bm{X}_i,\bm{X}')\leq S(\bm{X}'',\bm{X}'''))-\mu_x\bigr) + o_{\text{p}}(1),\\
\sqrt{n}(m^{(2)}_y-\mu_y) &= \frac{4}{\sqrt{n}} \sum_{i=1}^{n} \bigl(\P_{\tiny \bm{Y}',\bm{Y}'',\bm{Y}'''|\bm{Y}}(S(\bm{Y}_i,\bm{Y}')\leq S(\bm{Y}'',\bm{Y}'''))-\mu_y\bigr) + o_{\text{p}}(1),\\
  \sqrt{n}(m^{(2)}_{xy}-\mu_{xy}) &=\frac{4}{\sqrt{n}} \sum_{i=1}^{n} \bigl(\P_{\tiny \bm{X}',\bm{X}'',\bm{X}''',\bm{Y}',\bm{Y}'',\bm{Y}'''|\bm{X},\bm{Y}}(S(\bm{X}_i,\bm{X}')\leq S(\bm{X}'',\bm{X}'''),\\
  &\phantom{\,=\frac{4}{\sqrt{n}} \sum_{i=1}^{n} \bigl(}S(\bm{Y}_i,\bm{Y}') \leq S(\bm{Y}'',\bm{Y}'''))-\mu_{xy}\bigr) + o_{\text{p}}(1),
\end{align*}
where the conditional probabilities in the expressions above represent the
first order Hoeffding decomposition of the corresponding U-statistic. Combining all the terms, it follows that
\begin{align*}
\sqrt{n} \begin{pmatrix}
m^{(2)}_x -\mu_x\\ m^{(2)}_y-\mu_y \\ m^{(2)}_{xy}-\mu_{xy}
\end{pmatrix} \overset{\text{d}}{\longrightarrow} \N_3 (\bm{0},16 \Sigma_2),
\end{align*}
where $\Sigma_2$ denotes the covariance matrix of the random vector defined in
\eqref{tau:case2:randomvec}. One can then construct an estimator using
$\tau_n(S(\bm{X}),S(\bm{Y}))=f(m^{(2)}_x,m^{(2)}_y,m^{(2)}_{xy})$ where $f$ is defined as in the claim
Using the delta method, the convergence result follows.
\end{proof}

\begin{remark}
  In the U-statistics framework, one usually works with symmetric kernels as noted in \cite[Chapter~1]{lee1990}. For choices of collapsing functions which would yield non-symmetric kernels, note that one can easily replace the kernel with a symmetric variant. Suppose for example $\phi(X_1,\dots,X_m)$ is a kernel of order $m$. Then, the symmetric variant can be constructed as
  \begin{align*}
    \phi(X_1,\dots,X_m)=\frac{1}{m!} \sum_{\alpha_1,\dots,\alpha_m} \phi(X_{\alpha_1},\dots,X_{\alpha_m}),
  \end{align*}
  where the summation is taken over all permutations $(\alpha_1,\dots,\alpha_m)$ of $(1,\dots,m)$.
  By replacing any non-symmetric kernel with its symmetric variant, the rest of the derivation for the asymptotic distribution would then naturally follow.
\end{remark}
\end{document}

%
%
%
%

%%% Local Variables:
%%% mode: latex
%%% TeX-master: t
%%% End: